\documentclass[aps,pre,nofootinbib,amsfonts,superscriptaddress,longbibliography,showkeys,notitlepage]{revtex4-1}

\usepackage{braket}
\usepackage{hyperref}
\usepackage{graphicx,bm,amsmath,color}
\usepackage{hyperref}
\usepackage[flushleft]{threeparttable}
\usepackage[caption=false]{subfig}
\usepackage[gen]{eurosym}
\usepackage{amsthm}
\usepackage{float}
\usepackage{placeins}

\newtheorem{lemma}{Lemma}

\begin{document}

\title{Trading in Complex Networks}

\author{Felipe M. Cardoso}
\email{fmacielcardoso@gmail.com}
\affiliation{Institute for Biocomputation and Physics of Complex Systems,
	Universidad de Zaragoza, Spain}
\affiliation{Unidad Mixta Interdisciplinar de Comportamiento y Complejidad Social (UMICCS), UC3M-UV-UZ, Spain}
\author{Carlos Gracia-L\'azaro}
\affiliation{Institute for Biocomputation and Physics of Complex Systems,
Universidad de Zaragoza, Spain}
\affiliation{Unidad Mixta Interdisciplinar de Comportamiento y Complejidad Social (UMICCS), UC3M-UV-UZ, Spain}
\affiliation{Department of Theoretical Physics, Faculty of Sciences, Universidad de Zaragoza, Spain}
\author{Frederic Moisan}
\affiliation{Faculty of Economics, Cambridge University}
\author{Sanjeev Goyal}
\affiliation{Faculty of Economics and Christ's College, Cambridge University}
\author{\'Angel S\'anchez}
\affiliation{Department of Theoretical Physics, Faculty of Sciences, Universidad de Zaragoza, Spain}
\affiliation{Grupo Interdisciplinar de Sistemas Complejos, Departamento de Matemáticas, Universidad Carlos III de Madrid, 28911 Leganés, Madrid, Spain}
\affiliation{Institute UC3M-BS for Financial Big Data (IBiDat), Universidad Carlos III de Madrid, 28903 Getafe, Madrid, Spain}
\author{Yamir Moreno}
\email{yamir.moreno@gmail.com}
\affiliation{Institute for Biocomputation and Physics of Complex Systems,
Universidad de Zaragoza, Spain}
\affiliation{Unidad Mixta Interdisciplinar de Comportamiento y Complejidad Social (UMICCS), UC3M-UV-UZ, Spain}
\affiliation{Department of Theoretical Physics, Faculty of Sciences,
Universidad de Zaragoza, Spain}
\affiliation{ISI Foundation, Turin, Italy}

\begin{abstract}
Global supply networks in agriculture, manufacturing and services are a defining feature of the modern world. The efficiency and the distribution of surpluses across different parts of these networks depend on choices of intermediaries. This paper conducts price formation experiments with human subjects located in large complex networks to develop a better understanding of the principles governing behavior. Our first finding is that prices are larger and that trade is significantly less efficient in small world networks as compared to random networks. Our second finding is that location within a network is not an important determinant of pricing. An examination of the price dynamics suggests that traders on cheapest -- and hence active -- paths raise prices while those off these paths lower them. We construct an agent based model (ABM) that embodies this rule of thumb. Simulations of this ABM yield macroscopic patterns consistent with the experimental findings. Finally, we extrapolate the ABM on to significantly larger random and small world networks and find that network topology remains a key determinant of pricing and efficiency.
\end{abstract}

\maketitle

Globalization is a prominent feature of the modern economy \cite{rodrik2011globalization}. Nowadays, supply, service and trading chains \cite{anderson2004trade,antras2013organizing,hummels2001nature,antras2013organizing,
antras2011intermediated,chaabane2012design} play a central role in different contexts such as agriculture
\cite{fafchamps1999relationships,bayley2000revolution,meerman1997reforming,traub2008effects}, transport
and communication networks \citep{d2006antitrust,d2009asymmetry}, international trade \cite{dicken2001chains} and finance \cite{li2013dealer,gale2009trading}. One key question on these systems is how pricing dynamics by intermediaries of the economy impacts both efficiency and surpluses. The purpose of this paper is to develop a better understanding of the forces that shape intermediary pricing behavior in such complex networks.

Game theory constitutes a useful framework to study competition among trading agents \cite{osborne1990bargaining}. In this context, the Nash Bargaining Game \cite{nash1950bargaining} studies how two agents share a surplus that they
can jointly generate. In the Nash Bargaining Game, two players demand a portion of some
good. If the total amount requested by both players is less than the total value of the good, both players get
their request; otherwise, no player gets their request. There are many Nash equilibria in this game: any combination of demands whose sum is equal to the total value of the good constitutes a Nash equilibrium. There is also a Nash equilibrium where each player demands the entire value of the good \cite{nash1953two}. As a generalization of bargaining games to \textit{n} players, Choi \textit{et al.} \cite{Choi2017} proposed and tested in the laboratory a model of intermediation pricing.
In this model, a good is supposed to go from a source S to a destination D. Intermediaries,
which are located in the nodes of a network, may post a price for the passage of the good. Trading occurs if there exists a path between S and D on which the sum of prices is smaller than or equal to the value of the good.
The key finding was that the pricing and the surpluses of the intermediaries depends on the presence of
`critical' nodes: a node is said to be critical if it lies on all possible paths between S and D. They showed that intermediaries extract most of the surplus if and only if there exist critical nodes in the network. However, in all cases, intermediaries set prices so as to ensure that trade did take place. Efficiency was close to one hundred percent. These experiments were done in small groups
and relatively simple networks of up to 10 subjects. The goal of the present paper is to study trading in large scale networks  \cite{boccaletti2006complex}, to gain a better understanding into the role of network topology in commerce.

We conduct experiments with human subjects embedded in complex networks: specifically, we consider a random network and a small world network each with 26 subjects (and the same level of average connectivity). In these networks there are \textit{no} critical nodes: the results of Choi et al. would suggest that intermediary prices must be close to zero and that their surpluses must also be close to zero. As we will show below, our \textit{first} finding clearly rejects this conjecture: in all the networks studied, intermediaries set positive prices and they make large profits. Moreover, network topology has powerful effects: in particular, in the random network,  intermediaries set lower prices as compared to a small world network. As a consequence, there is full trading efficiency in the random network, but trade breaks down in almost one third of the cases in the small world network.

This striking difference leads us to an examination of how location within a network affects pricing: our \textit{second} finding is that within a given network, standard measures of network centrality appear to have no significant effect on pricing behavior. As network location does not matter for prices, the presence on the cheapest and active path must be crucial for profits. And indeed, this is what we observe: intermediaries' earnings are positively related to their betweenness weighted by the path length.

Turning to the dynamics of price setting, we observe that traders raise prices if they lie on the successful trade path (\textit{i.e.}, the least-cost path), and that they lower prices when they are off the least-cost path. Based on these observations, we build an agent based (ABM) model that reproduces qualitatively the experimental results. We then use simulations to extrapolate our findings to larger networks: our \textit{third} finding is that network topology continues to matter and that random networks exhibit lower prices and higher level of efficiency even when there are 100 traders. Finally, our \textit{forth} finding uncovers the role of node-disjoint paths $-$ two paths are disjoint if and only if they do not share any node$-$and of the average path length in shaping level of pricing: networks that have a larger number of node-disjoint paths exhibit lower prices and higher efficiency. Among networks with the same number of node-disjoint paths, average path length is an important driver of costs.

\section*{Experimental setup}

We consider a simple game of price setting in networks to study supply, service and trading chains taken from
Choi \textit{et al} \cite{Choi2017}. Let $\mathcal{N}$ be a set of nodes $\mathcal{N}= \lbrace S, D, 1, 2, \ldots, n\rbrace$, where $S$ is a source and $D$ a destination; and $\mathcal{L}$ a set of pairs of elements of $\mathcal{N}$.
$\mathcal{N}$ and $\mathcal{L}$ define a trading network where the elements of $\mathcal{L}$ are the links.
A path between $S$ and $D$ is a sequence of distinct nodes $\lbrace i_1, i_2, ..., i_l \rbrace$ such that
$\lbrace (S,i_1),(i_1,i_2), (i_2,i_3), \ldots, (i_{l-1},i_l),(i_l,D) \rbrace \subset  \mathcal{L}$.

Each experiment consists of 4 series of 15 rounds each, and it involves $n$ human subjects that will
play the role of intermediaries. Before starting the first round of a series, each
subject is randomly assigned to a node in $\lbrace 1, 2, \ldots, n\rbrace$. The positions
of $S$ and $D$ are also assigned at random. These positions (players, $S$ and $D$) remain constant over the 15 rounds.
Subjects are always informed about the network and their position in it,
that is, they can see the whole network including $S$ and $D$. At each round, every subject has to make a decision; namely, she has to post a price from 0 to 100 tokens for the passage of a good by her node.
The prices determine a total cost for every path between $S$ and $D$. A path is feasible if its cost is not greater than a given threshold (100 tokens) that represents the value of the economic good generated by the path. After all players have made their choices, the cheapest path is selected: if it is feasible, each player located in this path receives her proposed price as a payoff. Otherwise, no trade takes place and payoffs are zero. Players who are out of the selected path do not get any payoff in that round. In the case of more than one cheapest path, the tie is resolved through a random choice among cheapest paths. From the second round onward, players are informed about the existence of a trade in the previous round, about the previously selected path, and about the prices and payoffs of all the players in the previous round together with their positions in the network.

We have conducted two experimental sessions in a random network of 26 nodes with $\langle k \rangle =3$ and two more sessions in a small-world-like network of 26 nodes with $\langle k \rangle =3$. Additionally, we have conducted another experimental session in a random network of 50 nodes with $\langle k \rangle =4$ that will allow us to check the robustness of the results against the size and connectivity of the network. These results can be found in SI sections 1-2 and SI Table S2. Further details can be found in Materials and Methods section.

\section*{Results and discussion}

The networks used in the experiment allow for coexisting paths with a different number of intermediaries, where theory predicts both efficient and inefficient (Nash) equilibria. Furthermore, these networks present different characteristics, such as degree and centrality distributions, that may affect the bargaining power of the intermediaries. These facts motivate our first question: How does the network topology affect costs and prices? Fig. \ref{fig:fig1}A shows the cheapest path cost in each of the networks considered. As shown, the small-world networks exhibit higher costs than random networks (t(232.41)=15.5, $p<0.001$). Fig. \ref{fig:fig1}B instead displays the costs of the cheapest path normalized by the number of nodes on it, \textit{i.e.}, the mean price of nodes along the cheapest path. The differences between networks persist, indicating that prices and costs strongly depend on the topology of the network. These results, separated by rounds, are shown in Fig. S4 of the SI. Table 1 shows that there is a very large effect of topology on efficiency: in the random network trade is realized in practically all the cases, while in the small-world network trade breaks down in almost one third of the cases (binomial-test, 0.95 CI=(0.76, 0.90), $p<0.001$). However, small-world networks involve higher costs and profits than random networks, since the higher posted prices compensate for the lower efficiency. Therefore, we conclude that the topology of the network matters for intermediation: the surpluses of the intermediaries vary significantly from one network class to the other.

\begin{figure}[tbhp]
\centering
\includegraphics[width=.6\columnwidth]{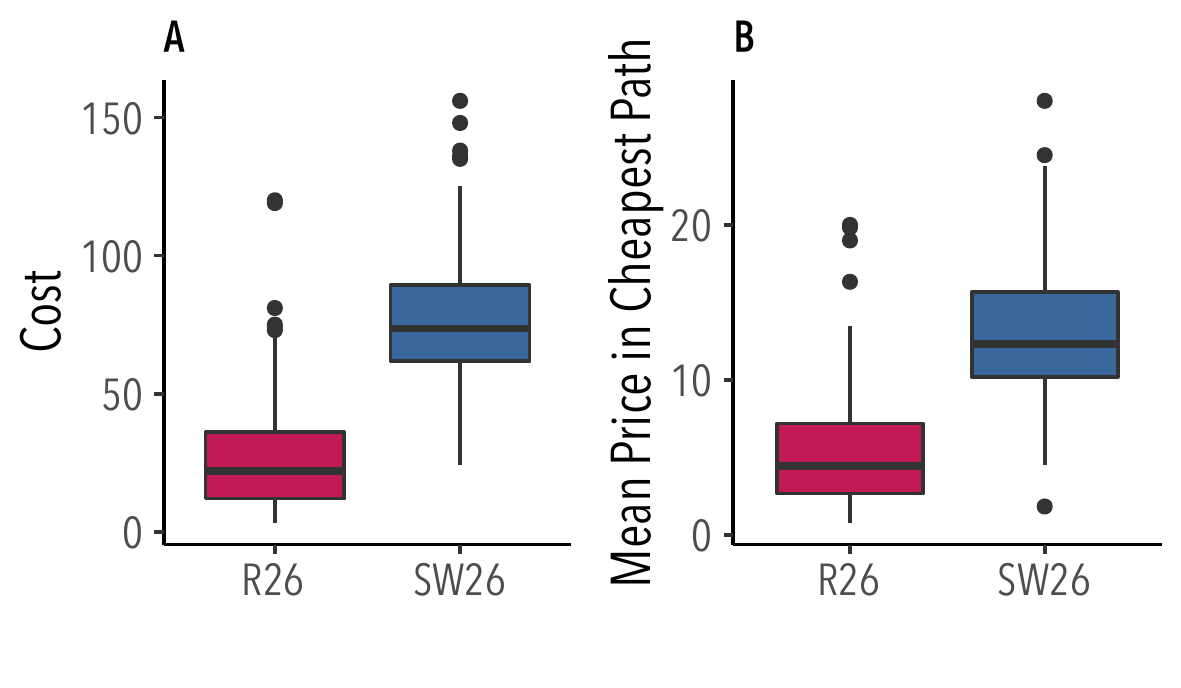}
\caption{
\textbf{Network topology affects trading costs and prices.}
\textbf{A}: Cost of the cheapest path for the random network of 26 nodes (R26), and for the small-world
network of 26 nodes (SW26).
 \textbf{B}: Mean price of participants in cheapest path
for the same networks.
Lines in the boxes denote the medians, whereas boxes extend to the lower- and upper-quartile values. Whiskers extend to the most extreme
values within 1.5 interquartile range (IQR).
}
\label{fig:fig1}
\end{figure}

\begin{table}[tbhp]
	\centering

	\begin{tabular}{lrrrrrr}
		\hline
		network & efficiency & price &  price in CP & cost & profit & length \\
		\toprule
		R 26 & 0.97 & 11.34 & 5.49 & 28.33 & 1.10 & 6.26 \\
		SW 26 & 0.68 & 18.10 & 13.16 &  76.52 & 2.38 & 7.00 \\
		\botrule
	\end{tabular}
	
	\caption{\textbf{Experimental results.}
	Efficiency (fraction of rounds
	in which the cheapest path cost was
	equal to or less than the threshold), and mean values of the price, price in the cheapest path, cost of the cheapest path,  profit, and cheapest path length
	for the random network with 26 nodes (R 26) and for the small-world network with 26 nodes
	(SW 26).}
	\label{table1}
\end{table}

Profit is only obtained when the subjects are on the cheapest path, i.e., when they are on the path through which the trading is realized. Thus, it is of interest to examine what is the role of the location of intermediaries in the network in shaping their behavior, which we do next. First, we observe that the networks in our experiment do not contain any critical nodes and yet they generate large rents. So the results from the small scale experiments by Choi \textit{et al} \cite{Choi2017} do not apply to complex larger scale networks. It seems likely then that nodes that are present on more paths have greater market power. This motivates a generalization of the notion of criticality as follows:

\begin{equation}\label{eq:sdc}
    sd(v)=\frac{|P_{SD}(v)|}{|P_{SD}|} \;\;,
\end{equation}
where $sd(v)$ is the \textit{partial} criticality of node $v$,
$|P_{SD}(v)|$ stands for the number of paths between the source and destination containing a given node $v$, and $|P_{SD}(v)|$ for the total number of paths between the source and the destination. Following this line of thought, a higher partial criticality may indicate a potential for greater bargaining power and therefore nodes with a higher partial criticality should show higher
prices and profits. Fig. \ref{fig:fig2}A shows the accumulated prices of the intermediaries
as a function of their partial criticality. There is no significant relation between partial criticality and the prices posted by participants. Even more strikingly, as  illustrated in Fig. \ref{fig:fig2}B, there is
no relationship between the accumulated payoff obtained and a node's partial criticality. Fig. \ref{fig:fig2}C shows the frequency that each player is on the cheapest path versus her partial criticality. Again, there is no relation between these variables.

\begin{figure*}
\centering
\includegraphics[width=.9\columnwidth]{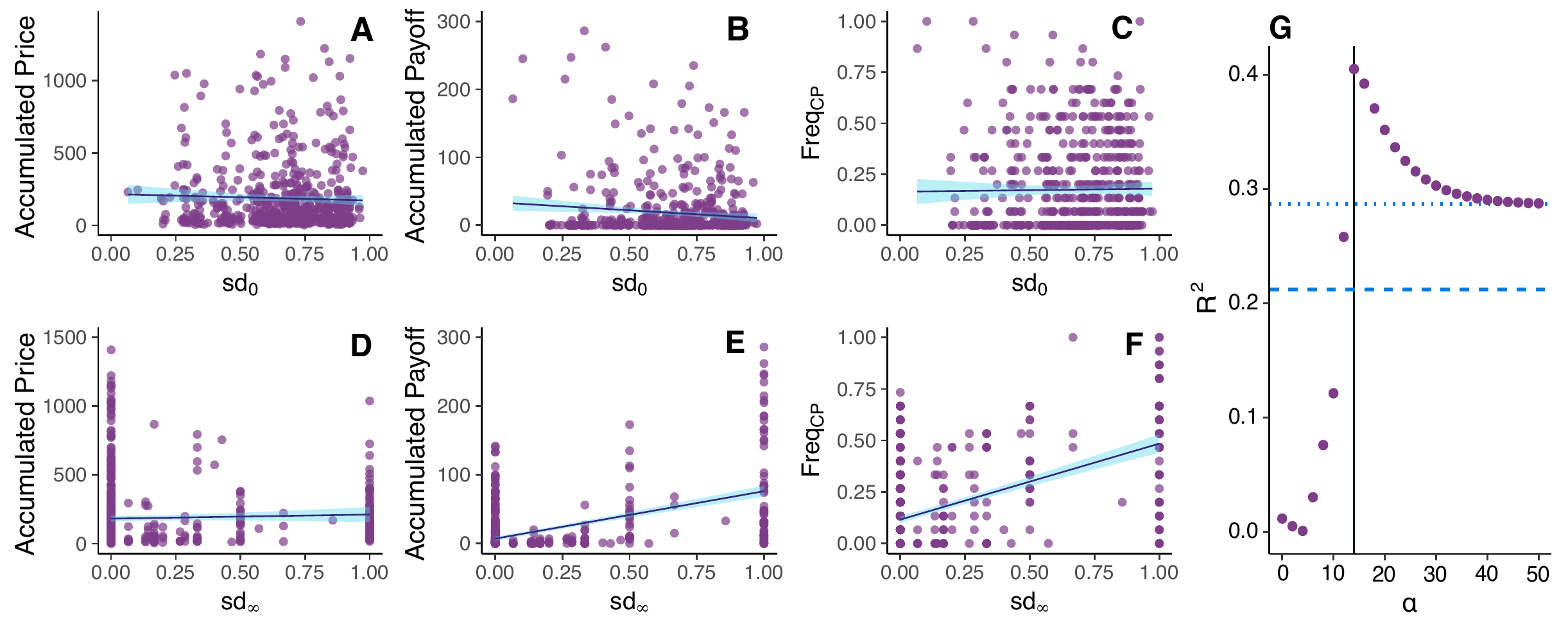}
\caption{
\textbf{SD-betweenness determines payoffs but not posted prices.}
 \textbf{A-F}: Accumulated price (\textbf{A,D}), accumulated payoff (\textbf{B,E}) and frequency in the cheapest path (\textbf{C,F}) of participants during a series of 15 rounds as a function of
the node criticallity $sd_0$ (\textbf{A,B,C}) and of
the SD-betweenness $sd_\infty$ (\textbf{D,E,F}).
\textbf{G:} $R^2$ of the regression of participants accumulated payoff
on $sd_\alpha$
versus $\alpha$, where
$\alpha$ modulates the weight of the length of the paths in the S-D centrality measure.
Dashed (points) line show the value of
$R^2$ for correlation of
payoffs on the betwenness (SD-betweenness). Data is pooled across any series of 15 rounds in any experimental session. For similar analyses within each experimental network, see SI Fig.S2.
}
\label{fig:fig2}
\end{figure*}

This lack of correlation may be due to the equal weighting of paths with different length. In
order to address this point, we refine our generalized notion of partial criticality to take path length into account:

\begin{equation}
     sd_\alpha(v)=\frac{\sum_{[S,v,D]} l(p)^{-\alpha}}{\sum_{[S,D]} l(p)^{-\alpha}} \; \; ,
     \label{eq:sdc.a}
\end{equation}

where the summations are over all the paths between $S$ and $D$ containing $v$ (numerator)
and over all the paths between $S$ and $D$ (denominator). $l(p)$ represents the length of
path $p$ and $\alpha$ stands for an arbitrary
weight: as $\alpha$ increases, more importance is given to shorter paths. Specifically, when $\alpha \to \infty$ it will consider only the shortest paths, $sd_\infty(v)$ being a measure
of the source-destination betweenness
of node $v$ (SD-betweenness$(v)$). On the opposite side, for $\alpha=0$ the
partial criticality of Equation \ref{fig:fig1}
is recovered, that is, $sd_0(v)=sd(v)$.

Fig. \ref{fig:fig2}D shows the accumulated prices of the intermediaries
as a function of their SD-betweenness. Again, there is no relation observed between pricing behavior
and betweenness. However, as shown in Fig. \ref{fig:fig2}E, there is a positive correlation between the accumulated payoff obtained by intermediaries $v$ and their $sd_\infty(v)$. That is, although pricing is uncorrelated with SD-betweenness centrality, profits are positively correlated with it. The reason behind this difference must therefore lie in how the presence of $v$ on the least-cost path is correlated with $sd_\infty(v)$. This is displayed in Fig. \ref{fig:fig2}F, which represents the fraction of times that
an intermediary is on the cheapest path versus her SD-betweenness. As shown, there is a positive correlation between these measures, which explains why -- in a situation where prices are largely insensitive to network location -- profits will be correlated with $sd_\infty(v)$. The robustness of these results against the size and connectivity of the network is discussed in the SI, Section 2.A.

So far, we have seen that node centrality does not influence earnings when we equally consider all the paths from S to D to compute it, but it does when we consider only the shortest paths. This fact indicates that the weight given to paths length is important to study the capacity of the nodes to extract surpluses. In order to verify this hypothesis, Fig. \ref{fig:fig2}G shows the coefficient of determination $R^2$ of the regression of intermediaries payoffs on $sd_\alpha$ as a function of $\alpha$. The best fit is obtained for $\alpha\sim 12$, which indicates that longer paths should have significantly smaller weight than shorter ones. As the number of paths grows exponentially with network size, SD-betweenness seems to be a feasible and good descriptor of participants' earnings.

\begin{figure}
\centering
\includegraphics[width=.6\columnwidth]{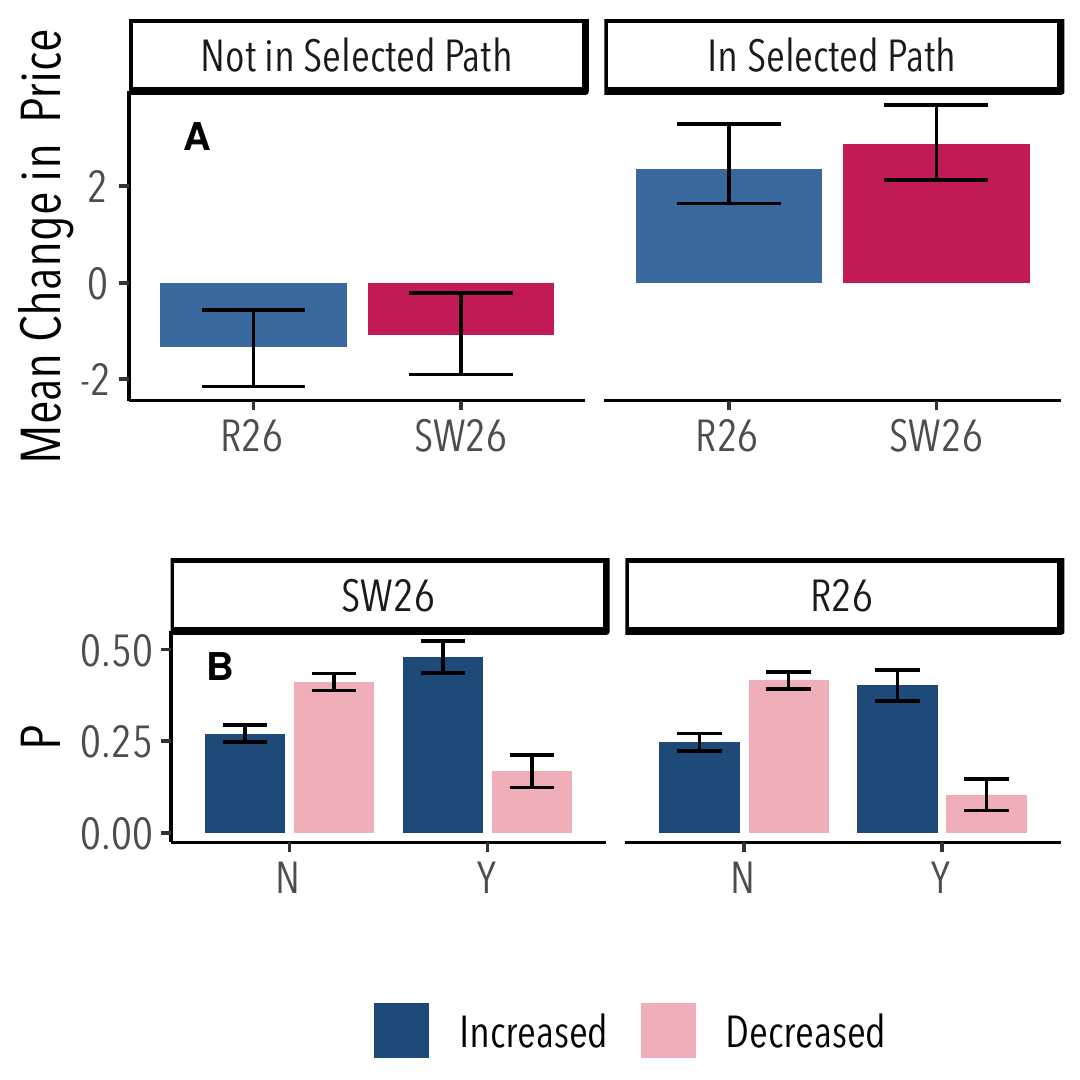}
\caption{
\textbf{Being or not in the cheapest path determines the intermediaries price increases.}
\textbf{A:} Mean changes in the posted price
conditioned to have been (right) or not (left)
in the selected cheapest path in the previous round
for the
random networks of 26 nodes (R26), and for the
small-world
network of 26 nodes (SW26).
\textbf{B:} Probability to increase
(red) and to decrease (blue) the posted price conditioned to have been (Y) or not (N) in the selected cheapest path, for each one  of the
studied networks. The error bars represent the 95$\%$ C.I. An extension of these results including the 50-nodes Random Network is displayed in SI Fig.S3.
}
\label{fig:fig3}
\end{figure}

\subsection*{Behavioral rules}

We have noted that participants' behavior is not determined by network position: criticality and classical measures of centrality are not good predictors of the prices posted by intermediaries. Nonetheless, results show differences in the prices posted by traders across different networks. Even if these networks might seem relatively small and similar, they are not. The environment (defined as the set of all the information that the individuals need to factor in their decisions) is very complex: there are many different paths passing through most of the traders, they need to take into account their price as well as those of other players, etc. It is thus reasonable to assume that the traders confronting such a complex and dynamic environment use rules of thumb, which on the other hand, should not depend on the network. In what follows, we develop a model that accounts for individual behavior and for the differences observed experimentally. 

Together with the network information, the other information shown to subjects is whether they were on the selected trading path.
Fig. \ref{fig:fig3}A shows, for each one of the networks considered,  the mean change in price for the cases when the participant was or was not along the cheapest path in the previous round. In the same way, Fig. \ref{fig:fig3}B shows the probabilities to increase and to decrease the posted price conditioned to have been (Y) or not (N) in the cheapest path. Players appear to follow a simple rule, namely, to increase their price if they were on the cheapest path in the previous round and to decrease it otherwise. Furthermore, the expected values shown in Fig. \ref{fig:fig3}A point out that successful intermediaries keep increasing their prices and therefore, without sufficient competition, costs and prices would always grow.

We now build a simple agent based model (ABM) \cite{Goldstone2005}, as described below:
\begin{itemize}
\item[i)] If node $u$ belongs to a cheapest path at time $t$, it will change its posted price on time $t+1$ by $\sigma$;

\item[ii)] If node $u$ does not belong to a cheapest path at time t, it will change its posted price on time $t+1$ by $\rho$;

\item[iii)] The minimum price a node can post is $0$.
\end{itemize}

\begin{figure}
\centering
\includegraphics[width=.6\columnwidth]{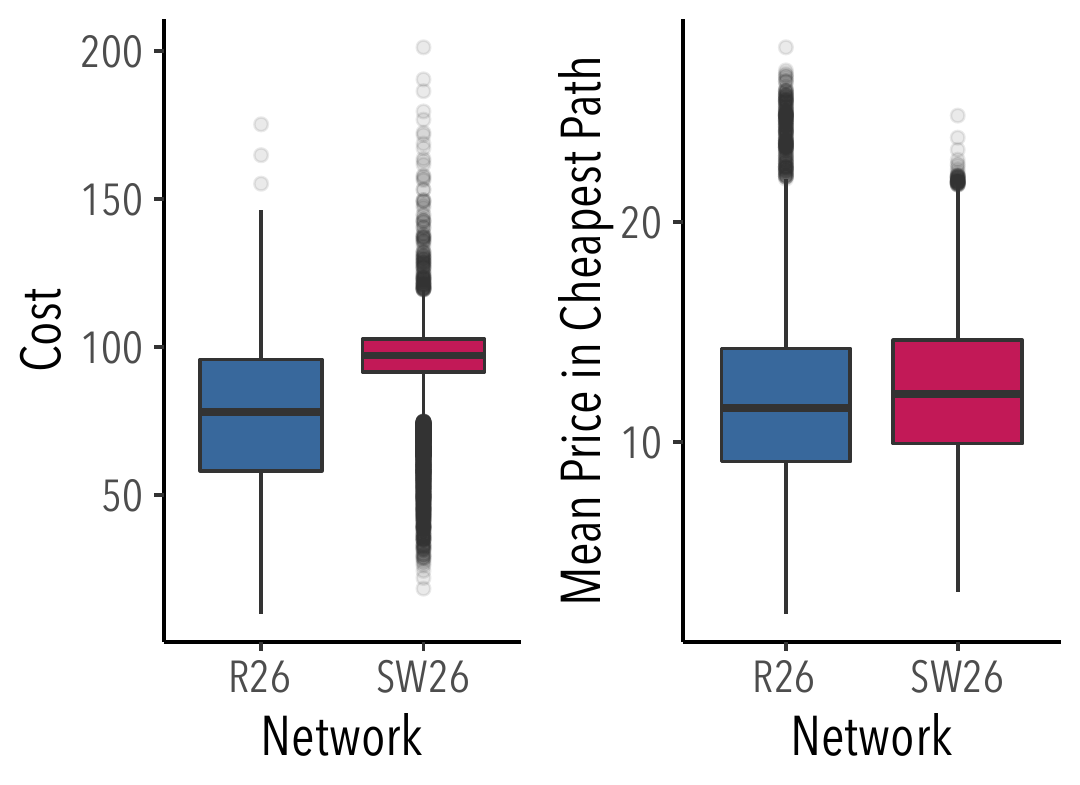}
\caption{\textbf{Numerical results of the model executed over the networks, source and destination from the experiments.} Results shown are for 100 executions with 15 rounds for each network and source-destination pair, excluding the first round. Initial prices are bootstrapped from the experimental values. Values of $\sigma$ and $\rho$ are fixed and correspond to to mean values of the experiment, respectively, $2.60$ and $1.2$.  \label{fig:newfig4}}
\end{figure}

To validate this model we executed it by bootstrapping the initial prices, the value of changes if on the cheapest path ($\sigma$) and the value of changes if not ($\rho$). The results, shown in Fig. \ref{fig:newfig4} and Table \ref{table2}, indicate that costs from simulations (resp. efficiency) are higher (resp. lower) in small-world networks than in random networks (t(9659.3)=68.33, $p<0.001$), in agreement with our experimental results. Costs reached relatively high values in some rounds, as the model does not incorporate participants direct response to the maximum cost threshold. Table  \ref{table2} confirms that topological differences between the networks are driving the differences in cost.

\begin{table}[tbhp]
\begin{tabular}{lrrrrr}
	\hline
	network & efficiency & price &  price in CP & cost & length \\
	\toprule
	R 26 & 0.87 & 15.14 & 11.87 & 75.68 & 7.62 \\
	SW 26 & 0.66 & 14.71 & 12.56 &  94.32 & 8.97 \\
	\botrule
\end{tabular}\\
\caption{\textbf{Numerical results in experimental networks.} Efficiency (fraction of rounds in which the cheapest path cost was equal to or less than the threshold), and mean values of the price, price in the cheapest path, cost of the cheapest path, and cheapest path length. Results obtained from numerical simulations with each one of the two studied networks with their corresponding source and destinations. }
\label{table2}
\end{table}

Once we have shown that the model captures very well the experimental observations, we can address our following objective. Specifically, we can use the model to extrapolate the behavior observed in the experiments to more general situations, such as other network sizes and longer times. This will allow us to uncover the role of the network properties affecting trading costs. To this end, we have performed numerical simulations of the model in networks as those used in the experiments as well as in larger ones. Results for networks of size 50 and 100 are shown in Fig. \ref{fig:fig5} and Table \ref{table3}, and they are also consistent with the experimental data. These results confirm that the network topology has a significant effect on trading outcomes: small worlds lead to higher costs and lower efficiency.

\begin{figure}
\centering
\includegraphics[width=.7\columnwidth]{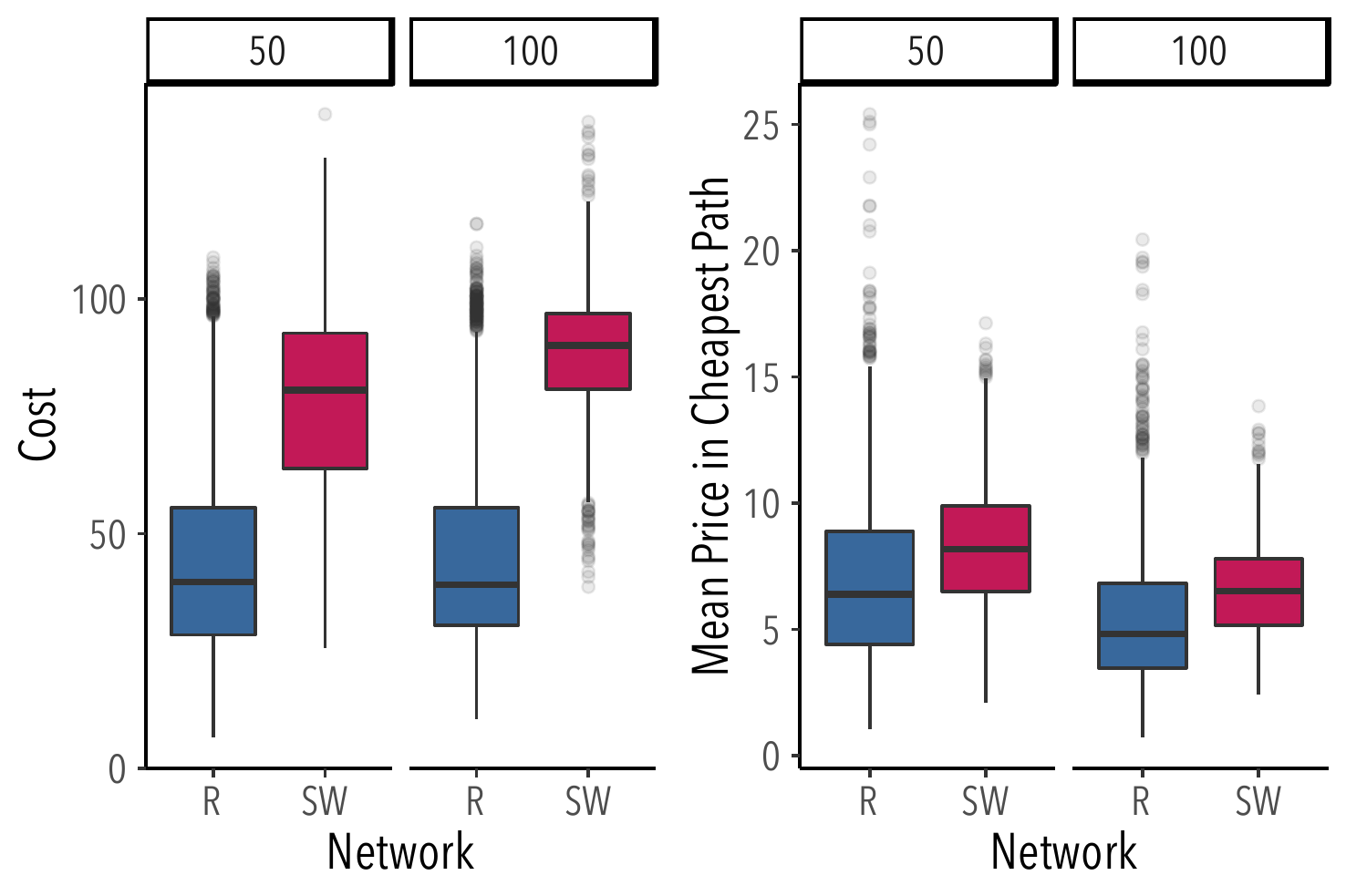}
\caption{\textbf{Numerical results of the model executed over the networks, sources and destinations used in the experiments.} Results shown are for 100 executions with 15 rounds for each network and source-destination pair, excluding the first round. Initial prices are bootstrapped from the experimental values. Values of $\sigma$ and $\rho$ are fixed and correspond to mean values of the experiment, respectively, $2.60$ and $1.2$. \label{fig:fig5}}
\end{figure}

\begin{table}[tbhp]

\begin{tabular}{lrrrrr}
	\hline
network & efficiency & price &  price in CP & cost & length \\
\toprule
R 50 & 0.98 & 13.12 & 7.03 & 44.76 & 7.68 \\
SW 50 & 0.91 & 14.32 & 8.28 &  77.44 & 10.74 \\
R 100 & 0.97 & 12.65 & 5.53 & 46.50 & 10.05 \\
SW 100 & 0.82 & 13.20 & 6.56 & 88.56 & 15.41 \\

\botrule
\end{tabular}
\caption{\textbf{Numerical results for larger networks.} Efficiency (fraction of rounds in which the cheapest path cost was equal to or less than the threshold), and mean values of the price, price in the cheapest path, cost of the cheapest path, and cheapest path length. Results obtained from numerical simulations with random networks with 50 and 100 nodes (R 50, R 100) for the small-world network with 50 and 100 nodes (R50, R 100).}
\label{table3}
\end{table}

Finally, we go one step further in order to explain what lies behind the differences found in costs. One possible theoretical hypothesis could be that costs depend on competition between paths. In our setup, this would be equivalent to assume that costs should decrease with the number of possible ways to reach the destination, \textit{i.e.}, the number of independent (sets of) paths from S to D. Specifically, we expect competition to  be proportional to the number $M$ of node-disjoint paths \cite{VanSteen2010}, as it captures the possible number of simultaneous independent trades\footnote{See SI Section 2.C.1 for a deeper discussion on this subject.}. According to this hypothesis, the larger the value of $M$, the lower the cost. Another possible explanation for the dependency of costs with the networks could be the structural differences between the latter. It is well known that clustering coefficients and average path lengths differ for the SW and the random networks considered in our experiments ($p \in \{0.1, 1\}$  \cite{watts1998collective}), and therefore the observed differences in cost could be tied to variations in those properties. 

In order to verify the previous hypotheses, we executed a version of the model without the maximum cost threshold. With this setup, we can study long-term effects after a sufficiently large number of rounds and uncover the cost tendency \footnote{In this regime, we cannot analyze network efficiency, however, networks yielding higher cost should be more inefficient.}. We ran the algorithm for $10^4$ rounds and then we considered the final cost of the trade for each configuration. Results for networks of size 26, 50 and 1000 nodes are shown in Figures \ref{fig:fig6}A,  \ref{fig:fig6}B, and  \ref{fig:fig6}C, respectively.

Simulations of trading dynamics on the aforementioned networks indicate that the number of node-disjoint paths ($M$) between $S$ and $D$ is the best indicator of final cost. Fig. \ref{fig:fig6}D shows that as $M$ grows, the costs are reduced so drastically that they go to 0 for  $M>3$. Moreover, the numerical results also reveal that for networks with the same value of $M$, the cost grows with the average path length. Indeed, this dependency explains why costs on small-world networks tend to be larger: these networks have a larger average path-length. To show that this finding is not a consequence of differences in the length of the cheapest paths, Fig. S5 of the SI displays, for the same simulations, the costs of the cheapest path normalized by the number of nodes on it versus the average path length of the network. It can be seen that the mean price of nodes in the cheapest path also correlates with the average path length. Interestingly, even though in this regime the difference in the clustering is larger than the difference in average path length, the former is not a good indicator of costs ($R^2=0.57$ vs $R^2=0.79$, see Section 2.C.2, table S3, and Figure S6 in the SI). In summary, these results provide two stylized facts that may guide future inquiries in this line, namely, trading costs will be null in setups with a relatively large number of node-disjoint paths and costs should be larger in networks with larger average path length. 

\begin{figure}
\centering
\includegraphics[width=.6\columnwidth]{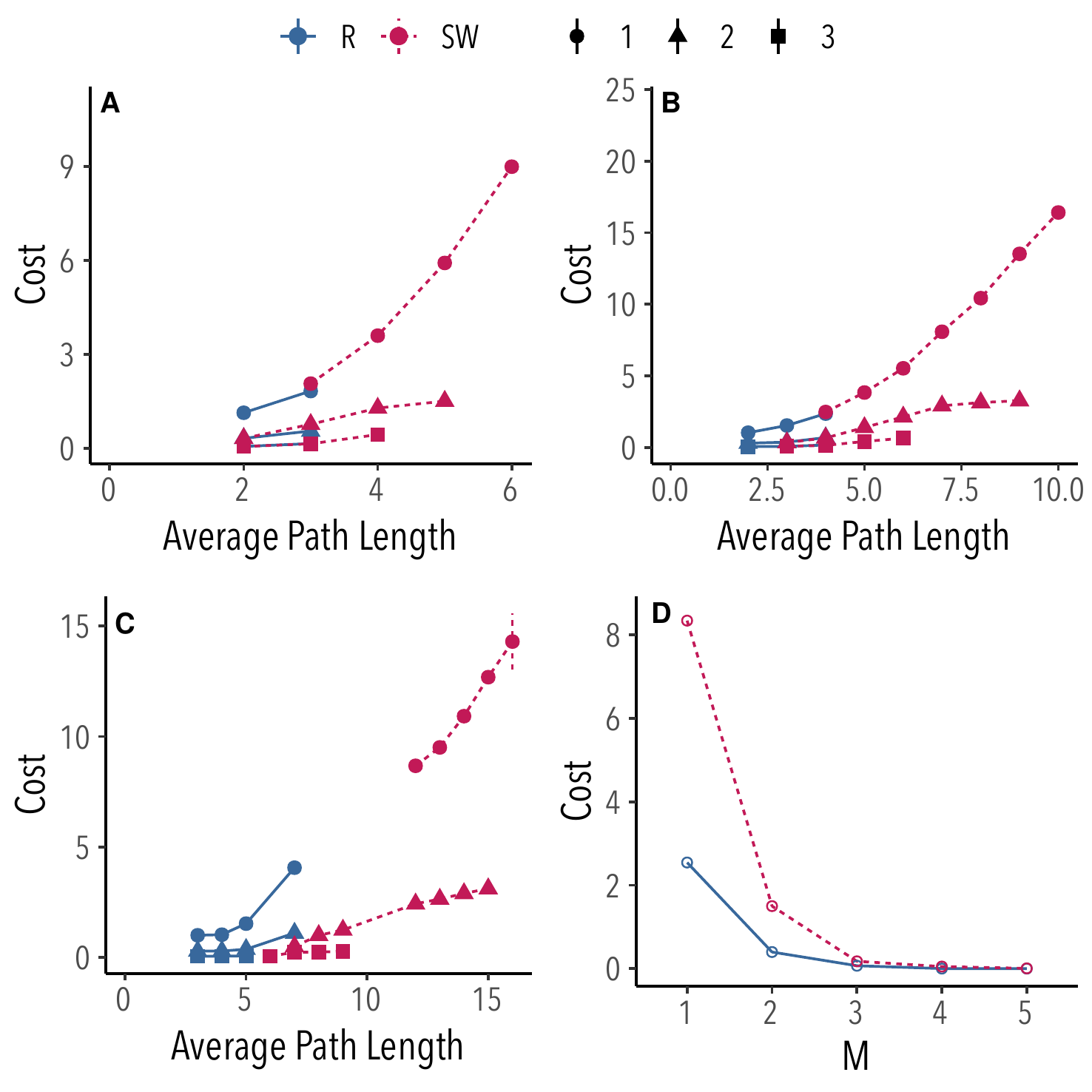}
\caption{
\textbf{Numerical results of the model. }\textbf{A,B,C}:
Average final cost (in $10^4$) of the cheapest path after a period of $10^4$ rounds as a function of the average path length of the network. Different panels correspond to different network sizes: 26 (\textbf{A}), 50 (\textbf{B}),
and 1000 (\textbf{C}) nodes; colors correspond to different network models: random (\textit{red}) and small-world (\textit{blue}); and different shapes correspond to different values of the number $M$ of disjoint paths.
For each configuration, there were generated 10000 networks of size 26,  50, and 1000, according
to the Watts-Strogatz algorithm \cite{watts1998collective}
with $p=0.1, 1$ and average degree from 2 to 10.
The initial cost was set to 0 and the increment/decrement ratio was
fixed to the experimental value
($\sigma/\rho=2.4$). Results for $M > 5$ are not shown as costs converge fast to 0.
\textbf{D:} Mean value of the cost of
the cheapest path versus $M$
for the same networks. See the main text
for further details.}
\label{fig:fig6}
\end{figure}

\section*{Conclusions}
Our experimental results indicate that the topological properties of a network have a powerful effect on pricing behavior of intermediaries and on the overall efficiency of the system. However, within a network, prices are relatively insensitive to node location, but intermediaries with greater betweenness make larger profits. Informed by the experimental results, we introduced an ABM of pricing behavior to understand traders' pricing. The key input of the model is the experimental observation that intermediaries raise prices when they lie on the cheapest path and lower their prices otherwise. The model successfully reproduced qualitatively the experimental results and allowed us to extrapolate and anticipate outcomes of pricing and efficiency to scenarios involving larger networks and longer timescales. Important enough, the model also enabled the discovery of what are the key determinants of cost, namely, the number of node-disjoint paths from source to destination and the network average path length. Ultimately, this explained the differences in our experimental results: in a small world network, the average path length tends to be larger and this leads to higher costs and lower efficiency of trading in these networks as compared to random networks.

Overall, our work reveals that the topology of trading networks is key to determine their efficiency and cost. It would be interesting to further test our conclusions using real data on trading, in particular, the finding that the availability of node-disjoint paths takes trading costs down. On the other hand, our insights may be useful for the design of competition-improved networks for goods currently overpriced due to intermediation. Further research on the role of information provided to intermediaries and on other network topologies will be also relevant to address these issues. 

\section*{Materials and methods}
We carried out 5 experimental sessions, each of which is composed of 4 consecutive series of 15 rounds each. All the networks have been generated through
the Watts-Strogatz algorithm \cite{watts1998collective} with different probabilities $p$ of rewiring ($p = 0.1$ for small-world-like networks and $p = 1$ for random networks). Mean degree was $\langle k\rangle =3$
for the 26 nodes networks and $\langle k\rangle =4$ for the network of 50 nodes. A representation of each network as viewed by the participants is shown in SI Fig. S1.
The experiment was conducted with 144 volunteers recruited from the volunteer pool of
the IBSEN project (http://www.ibsen.eu). The first experimental session was performed on 29 June 2017, jointly at Experimental Economics Labs of University of Zaragoza (UZ) and University Carlos III of Madrid (UC3M), with a random network of 50 nodes. A second experimental session was performed on 31 October 2017, at UC3M, with a SW network network of 26 nodes. Subsequent three sessions were performed at UZ on 14 November 2018 (RN, 26 nodes), 24 April 2018 (RN, 26 nodes), and 25 April 2018 (SW, 26 nodes). Table S1 of the SI shows the demographic data of the sessions.

All the participants played through a web interface after reading a tutorial (both included in SI Materials and Methods) on the screen. When everybody had gone through the tutorial, the experiment began, lasting for approximately 90 minutes. At the end of the experiments, all participants received their earnings and their show-up fee. Total earnings in the experiment ranged from 5 to 82 Euros, average earning was 18.4 Euros.

All participants signed an informed consent to participate. Besides, their anonymity was always preserved (in agreement with the Spanish Law for Personal Data Protection) by assigning them randomly a username which identified them in the system. No association was ever made between their real names and the results. This procedure was checked and approved by the Clinical Research Ethical Committee of IACS, Aragon Government.

\begin{acknowledgements}All authors have been partially supported by the EU through FET Open Project IBSEN (contract no. 662725). A.S.  acknowledges partial support by MINECO/FEDER (Spain) through grant FIS2015-64349-P VARIANCE (AS). Y.M. acknowledges partial support from the Government of Arag\'on, Spain through a grant to the group FENOL (E36-17R), by MINECO and FEDER funds (grant FIS2017-87519-P) and by Intesa Sanpaolo Innovation Center. The funders had no role in study design, data collection and analysis, or preparation of the manuscript.
\end{acknowledgements}

\bibliography{Trading_PNAS.bib}

\begin{thebibliography}{24}%
\makeatletter
\providecommand \@ifxundefined [1]{%
 \@ifx{#1\undefined}
}%
\providecommand \@ifnum [1]{%
 \ifnum #1\expandafter \@firstoftwo
 \else \expandafter \@secondoftwo
 \fi
}%
\providecommand \@ifx [1]{%
 \ifx #1\expandafter \@firstoftwo
 \else \expandafter \@secondoftwo
 \fi
}%
\providecommand \natexlab [1]{#1}%
\providecommand \enquote  [1]{``#1''}%
\providecommand \bibnamefont  [1]{#1}%
\providecommand \bibfnamefont [1]{#1}%
\providecommand \citenamefont [1]{#1}%
\providecommand \href@noop [0]{\@secondoftwo}%
\providecommand \href [0]{\begingroup \@sanitize@url \@href}%
\providecommand \@href[1]{\@@startlink{#1}\@@href}%
\providecommand \@@href[1]{\endgroup#1\@@endlink}%
\providecommand \@sanitize@url [0]{\catcode `\\12\catcode `\$12\catcode
  `\&12\catcode `\#12\catcode `\^12\catcode `\_12\catcode `\%12\relax}%
\providecommand \@@startlink[1]{}%
\providecommand \@@endlink[0]{}%
\providecommand \url  [0]{\begingroup\@sanitize@url \@url }%
\providecommand \@url [1]{\endgroup\@href {#1}{\urlprefix }}%
\providecommand \urlprefix  [0]{URL }%
\providecommand \Eprint [0]{\href }%
\providecommand \doibase [0]{http://dx.doi.org/}%
\providecommand \selectlanguage [0]{\@gobble}%
\providecommand \bibinfo  [0]{\@secondoftwo}%
\providecommand \bibfield  [0]{\@secondoftwo}%
\providecommand \translation [1]{[#1]}%
\providecommand \BibitemOpen [0]{}%
\providecommand \bibitemStop [0]{}%
\providecommand \bibitemNoStop [0]{.\EOS\space}%
\providecommand \EOS [0]{\spacefactor3000\relax}%
\providecommand \BibitemShut  [1]{\csname bibitem#1\endcsname}%
\let\auto@bib@innerbib\@empty
\bibitem [{\citenamefont {Rodrik}(2011)}]{rodrik2011globalization}%
  \BibitemOpen
  \bibfield  {author} {\bibinfo {author} {\bibfnamefont {Dani}\ \bibnamefont
  {Rodrik}},\ }\href@noop {} {\emph {\bibinfo {title} {The globalization
  paradox: democracy and the future of the world economy}}}\ (\bibinfo
  {publisher} {WW Norton \& Company},\ \bibinfo {year} {2011})\BibitemShut
  {NoStop}%
\bibitem [{\citenamefont {Anderson}\ and\ \citenamefont
  {Van~Wincoop}(2004)}]{anderson2004trade}%
  \BibitemOpen
  \bibfield  {author} {\bibinfo {author} {\bibfnamefont {James~E}\ \bibnamefont
  {Anderson}}\ and\ \bibinfo {author} {\bibfnamefont {Eric}\ \bibnamefont
  {Van~Wincoop}},\ }\bibfield  {title} {\enquote {\bibinfo {title} {Trade
  costs},}\ }\href@noop {} {\bibfield  {journal} {\bibinfo  {journal} {Journal
  of Economic literature}\ }\textbf {\bibinfo {volume} {42}},\ \bibinfo {pages}
  {691--751} (\bibinfo {year} {2004})}\BibitemShut {NoStop}%
\bibitem [{\citenamefont {Antr{\`a}s}\ and\ \citenamefont
  {Chor}(2013)}]{antras2013organizing}%
  \BibitemOpen
  \bibfield  {author} {\bibinfo {author} {\bibfnamefont {Pol}\ \bibnamefont
  {Antr{\`a}s}}\ and\ \bibinfo {author} {\bibfnamefont {Davin}\ \bibnamefont
  {Chor}},\ }\bibfield  {title} {\enquote {\bibinfo {title} {Organizing the
  global value chain},}\ }\href@noop {} {\bibfield  {journal} {\bibinfo
  {journal} {Econometrica}\ }\textbf {\bibinfo {volume} {81}},\ \bibinfo
  {pages} {2127--2204} (\bibinfo {year} {2013})}\BibitemShut {NoStop}%
\bibitem [{\citenamefont {Hummels}\ \emph {et~al.}(2001)\citenamefont
  {Hummels}, \citenamefont {Ishii},\ and\ \citenamefont
  {Yi}}]{hummels2001nature}%
  \BibitemOpen
  \bibfield  {author} {\bibinfo {author} {\bibfnamefont {David}\ \bibnamefont
  {Hummels}}, \bibinfo {author} {\bibfnamefont {Jun}\ \bibnamefont {Ishii}}, \
  and\ \bibinfo {author} {\bibfnamefont {Kei-Mu}\ \bibnamefont {Yi}},\
  }\bibfield  {title} {\enquote {\bibinfo {title} {The nature and growth of
  vertical specialization in world trade},}\ }\href@noop {} {\bibfield
  {journal} {\bibinfo  {journal} {Journal of international Economics}\ }\textbf
  {\bibinfo {volume} {54}},\ \bibinfo {pages} {75--96} (\bibinfo {year}
  {2001})}\BibitemShut {NoStop}%
\bibitem [{\citenamefont {Antras}\ and\ \citenamefont
  {Costinot}(2011)}]{antras2011intermediated}%
  \BibitemOpen
  \bibfield  {author} {\bibinfo {author} {\bibfnamefont {Pol}\ \bibnamefont
  {Antras}}\ and\ \bibinfo {author} {\bibfnamefont {Arnaud}\ \bibnamefont
  {Costinot}},\ }\bibfield  {title} {\enquote {\bibinfo {title} {Intermediated
  trade},}\ }\href@noop {} {\bibfield  {journal} {\bibinfo  {journal} {The
  Quarterly Journal of Economics}\ }\textbf {\bibinfo {volume} {126}},\
  \bibinfo {pages} {1319--1374} (\bibinfo {year} {2011})}\BibitemShut {NoStop}%
\bibitem [{\citenamefont {Chaabane}\ \emph {et~al.}(2012)\citenamefont
  {Chaabane}, \citenamefont {Ramudhin},\ and\ \citenamefont
  {Paquet}}]{chaabane2012design}%
  \BibitemOpen
  \bibfield  {author} {\bibinfo {author} {\bibfnamefont {Amin}\ \bibnamefont
  {Chaabane}}, \bibinfo {author} {\bibfnamefont {Amar}\ \bibnamefont
  {Ramudhin}}, \ and\ \bibinfo {author} {\bibfnamefont {Marc}\ \bibnamefont
  {Paquet}},\ }\bibfield  {title} {\enquote {\bibinfo {title} {Design of
  sustainable supply chains under the emission trading scheme},}\ }\href@noop
  {} {\bibfield  {journal} {\bibinfo  {journal} {International Journal of
  Production Economics}\ }\textbf {\bibinfo {volume} {135}},\ \bibinfo {pages}
  {37--49} (\bibinfo {year} {2012})}\BibitemShut {NoStop}%
\bibitem [{\citenamefont {Fafchamps}\ and\ \citenamefont
  {Minten}(1999)}]{fafchamps1999relationships}%
  \BibitemOpen
  \bibfield  {author} {\bibinfo {author} {\bibfnamefont {Marcel}\ \bibnamefont
  {Fafchamps}}\ and\ \bibinfo {author} {\bibfnamefont {Bart}\ \bibnamefont
  {Minten}},\ }\bibfield  {title} {\enquote {\bibinfo {title} {Relationships
  and traders in madagascar},}\ }\href@noop {} {\bibfield  {journal} {\bibinfo
  {journal} {The Journal of Development Studies}\ }\textbf {\bibinfo {volume}
  {35}},\ \bibinfo {pages} {1--35} (\bibinfo {year} {1999})}\BibitemShut
  {NoStop}%
\bibitem [{\citenamefont {Bayley}\ \emph {et~al.}(2000)\citenamefont {Bayley}
  \emph {et~al.}}]{bayley2000revolution}%
  \BibitemOpen
  \bibfield  {author} {\bibinfo {author} {\bibfnamefont {Brendan}\ \bibnamefont
  {Bayley}} \emph {et~al.},\ }\href@noop {} {\emph {\bibinfo {title} {A
  revolution in the market: the deregulation of South African agriculture.}}}\
  (\bibinfo  {publisher} {Oxford Policy Management},\ \bibinfo {year}
  {2000})\BibitemShut {NoStop}%
\bibitem [{\citenamefont {Meerman}(1997)}]{meerman1997reforming}%
  \BibitemOpen
  \bibfield  {author} {\bibinfo {author} {\bibfnamefont {Jacob}\ \bibnamefont
  {Meerman}},\ }\href@noop {} {\emph {\bibinfo {title} {Reforming agriculture:
  The World Bank goes to market}}}\ (\bibinfo  {publisher} {The World Bank},\
  \bibinfo {year} {1997})\BibitemShut {NoStop}%
\bibitem [{\citenamefont {Traub}\ and\ \citenamefont
  {Jayne}(2008)}]{traub2008effects}%
  \BibitemOpen
  \bibfield  {author} {\bibinfo {author} {\bibfnamefont {Lulama~Ndibongo}\
  \bibnamefont {Traub}}\ and\ \bibinfo {author} {\bibfnamefont {Thomas~S}\
  \bibnamefont {Jayne}},\ }\bibfield  {title} {\enquote {\bibinfo {title} {The
  effects of price deregulation on maize marketing margins in south africa},}\
  }\href@noop {} {\bibfield  {journal} {\bibinfo  {journal} {Food Policy}\
  }\textbf {\bibinfo {volume} {33}},\ \bibinfo {pages} {224--236} (\bibinfo
  {year} {2008})}\BibitemShut {NoStop}%
\bibitem [{\citenamefont {D'Ignazio}\ and\ \citenamefont
  {Giovannetti}(2006)}]{d2006antitrust}%
  \BibitemOpen
  \bibfield  {author} {\bibinfo {author} {\bibfnamefont {Alessio}\ \bibnamefont
  {D'Ignazio}}\ and\ \bibinfo {author} {\bibfnamefont {Emanuele}\ \bibnamefont
  {Giovannetti}},\ }\bibfield  {title} {\enquote {\bibinfo {title} {Antitrust
  analysis for the internet upstream market: A border gateway protocol
  approach},}\ }\href@noop {} {\bibfield  {journal} {\bibinfo  {journal}
  {Journal of Competition Law and Economics}\ }\textbf {\bibinfo {volume}
  {2}},\ \bibinfo {pages} {43--69} (\bibinfo {year} {2006})}\BibitemShut
  {NoStop}%
\bibitem [{\citenamefont {D'Ignazio}\ and\ \citenamefont
  {Giovannetti}(2009)}]{d2009asymmetry}%
  \BibitemOpen
  \bibfield  {author} {\bibinfo {author} {\bibfnamefont {Alessio}\ \bibnamefont
  {D'Ignazio}}\ and\ \bibinfo {author} {\bibfnamefont {Emanuele}\ \bibnamefont
  {Giovannetti}},\ }\bibfield  {title} {\enquote {\bibinfo {title} {Asymmetry
  and discrimination in internet peering: Evidence from the linx},}\
  }\href@noop {} {\bibfield  {journal} {\bibinfo  {journal} {International
  Journal of Industrial Organization}\ }\textbf {\bibinfo {volume} {27}},\
  \bibinfo {pages} {441--448} (\bibinfo {year} {2009})}\BibitemShut {NoStop}%
\bibitem [{\citenamefont {Dicken}\ \emph {et~al.}(2001)\citenamefont {Dicken},
  \citenamefont {Kelly}, \citenamefont {Olds},\ and\ \citenamefont
  {Wai-Chung~Yeung}}]{dicken2001chains}%
  \BibitemOpen
  \bibfield  {author} {\bibinfo {author} {\bibfnamefont {Peter}\ \bibnamefont
  {Dicken}}, \bibinfo {author} {\bibfnamefont {Philip~F}\ \bibnamefont
  {Kelly}}, \bibinfo {author} {\bibfnamefont {Kris}\ \bibnamefont {Olds}}, \
  and\ \bibinfo {author} {\bibfnamefont {Henry}\ \bibnamefont
  {Wai-Chung~Yeung}},\ }\bibfield  {title} {\enquote {\bibinfo {title} {Chains
  and networks, territories and scales: towards a relational framework for
  analysing the global economy},}\ }\href@noop {} {\bibfield  {journal}
  {\bibinfo  {journal} {Global networks}\ }\textbf {\bibinfo {volume} {1}},\
  \bibinfo {pages} {89--112} (\bibinfo {year} {2001})}\BibitemShut {NoStop}%
\bibitem [{\citenamefont {Li}\ and\ \citenamefont
  {Sch{\"u}rhoff}(2013)}]{li2013dealer}%
  \BibitemOpen
  \bibfield  {author} {\bibinfo {author} {\bibfnamefont {Dan}\ \bibnamefont
  {Li}}\ and\ \bibinfo {author} {\bibfnamefont {Norman}\ \bibnamefont
  {Sch{\"u}rhoff}},\ }\href@noop {} {\emph {\bibinfo {title} {Dealer networks:
  market quality in over-the-counter markets}}},\ \bibinfo {type} {Tech. Rep.}\
  (\bibinfo  {institution} {Mimeo University of Lausanne},\ \bibinfo {year}
  {2013})\BibitemShut {NoStop}%
\bibitem [{\citenamefont {Gale}\ and\ \citenamefont
  {Kariv}(2009)}]{gale2009trading}%
  \BibitemOpen
  \bibfield  {author} {\bibinfo {author} {\bibfnamefont {Douglas~M}\
  \bibnamefont {Gale}}\ and\ \bibinfo {author} {\bibfnamefont {Shachar}\
  \bibnamefont {Kariv}},\ }\bibfield  {title} {\enquote {\bibinfo {title}
  {Trading in networks: A normal form game experiment},}\ }\href@noop {}
  {\bibfield  {journal} {\bibinfo  {journal} {American Economic Journal:
  Microeconomics}\ }\textbf {\bibinfo {volume} {1}},\ \bibinfo {pages}
  {114--32} (\bibinfo {year} {2009})}\BibitemShut {NoStop}%
\bibitem [{\citenamefont {Osborne}\ and\ \citenamefont
  {Rubinstein}(1990)}]{osborne1990bargaining}%
  \BibitemOpen
  \bibfield  {author} {\bibinfo {author} {\bibfnamefont {Martin~J}\
  \bibnamefont {Osborne}}\ and\ \bibinfo {author} {\bibfnamefont {Ariel}\
  \bibnamefont {Rubinstein}},\ }\href@noop {} {\emph {\bibinfo {title}
  {Bargaining and markets}}}\ (\bibinfo  {publisher} {Academic press},\
  \bibinfo {year} {1990})\BibitemShut {NoStop}%
\bibitem [{\citenamefont {Nash~Jr}(1950)}]{nash1950bargaining}%
  \BibitemOpen
  \bibfield  {author} {\bibinfo {author} {\bibfnamefont {John~F}\ \bibnamefont
  {Nash~Jr}},\ }\bibfield  {title} {\enquote {\bibinfo {title} {The bargaining
  problem},}\ }\href@noop {} {\bibfield  {journal} {\bibinfo  {journal}
  {Econometrica: Journal of the Econometric Society}\ ,\ \bibinfo {pages}
  {155--162}} (\bibinfo {year} {1950})}\BibitemShut {NoStop}%
\bibitem [{\citenamefont {Nash}(1953)}]{nash1953two}%
  \BibitemOpen
  \bibfield  {author} {\bibinfo {author} {\bibfnamefont {John}\ \bibnamefont
  {Nash}},\ }\bibfield  {title} {\enquote {\bibinfo {title} {Two-person
  cooperative games},}\ }\href@noop {} {\bibfield  {journal} {\bibinfo
  {journal} {Econometrica: Journal of the Econometric Society}\ ,\ \bibinfo
  {pages} {128--140}} (\bibinfo {year} {1953})}\BibitemShut {NoStop}%
\bibitem [{\citenamefont {Choi}\ \emph {et~al.}(2017)\citenamefont {Choi},
  \citenamefont {Galeotti},\ and\ \citenamefont {Goyal}}]{Choi2017}%
  \BibitemOpen
  \bibfield  {author} {\bibinfo {author} {\bibfnamefont {Syngjoo}\ \bibnamefont
  {Choi}}, \bibinfo {author} {\bibfnamefont {Andrea}\ \bibnamefont {Galeotti}},
  \ and\ \bibinfo {author} {\bibfnamefont {Sanjeev}\ \bibnamefont {Goyal}},\
  }\bibfield  {title} {\enquote {\bibinfo {title} {Trading in networks: Theory
  and experiments},}\ }\href@noop {} {\bibfield  {journal} {\bibinfo  {journal}
  {Journal of the European Economic Association}\ }\textbf {\bibinfo {volume}
  {15}},\ \bibinfo {pages} {784--817} (\bibinfo {year} {2017})}\BibitemShut
  {NoStop}%
\bibitem [{\citenamefont {Boccaletti}\ \emph {et~al.}(2006)\citenamefont
  {Boccaletti}, \citenamefont {Latora}, \citenamefont {Moreno}, \citenamefont
  {Chavez},\ and\ \citenamefont {Hwang}}]{boccaletti2006complex}%
  \BibitemOpen
  \bibfield  {author} {\bibinfo {author} {\bibfnamefont {Stefano}\ \bibnamefont
  {Boccaletti}}, \bibinfo {author} {\bibfnamefont {Vito}\ \bibnamefont
  {Latora}}, \bibinfo {author} {\bibfnamefont {Yamir}\ \bibnamefont {Moreno}},
  \bibinfo {author} {\bibfnamefont {Martin}\ \bibnamefont {Chavez}}, \ and\
  \bibinfo {author} {\bibfnamefont {D-U}\ \bibnamefont {Hwang}},\ }\bibfield
  {title} {\enquote {\bibinfo {title} {Complex networks: Structure and
  dynamics},}\ }\href@noop {} {\bibfield  {journal} {\bibinfo  {journal}
  {Physics reports}\ }\textbf {\bibinfo {volume} {424}},\ \bibinfo {pages}
  {175--308} (\bibinfo {year} {2006})}\BibitemShut {NoStop}%
\bibitem [{\citenamefont {Goldstone}\ and\ \citenamefont
  {Janssen}(2005)}]{Goldstone2005}%
  \BibitemOpen
  \bibfield  {author} {\bibinfo {author} {\bibfnamefont {Robert~L.}\
  \bibnamefont {Goldstone}}\ and\ \bibinfo {author} {\bibfnamefont {Marco~A.}\
  \bibnamefont {Janssen}},\ }\bibfield  {title} {\enquote {\bibinfo {title}
  {Computational models of collective behavior},}\ }\href {\doibase
  https://doi.org/10.1016/j.tics.2005.07.009} {\bibfield  {journal} {\bibinfo
  {journal} {Trends in Cognitive Sciences}\ }\textbf {\bibinfo {volume} {9}},\
  \bibinfo {pages} {424 -- 430} (\bibinfo {year} {2005})}\BibitemShut {NoStop}%
\bibitem [{\citenamefont {van Steen}(2010)}]{VanSteen2010}%
  \BibitemOpen
  \bibfield  {author} {\bibinfo {author} {\bibfnamefont {Marteen}\ \bibnamefont
  {van Steen}},\ }\href@noop {} {\emph {\bibinfo {title} {{Graph Theory and
  Complex Networks}}}}\ (\bibinfo {year} {2010})\BibitemShut {NoStop}%
\bibitem [{\citenamefont {Watts}\ and\ \citenamefont
  {Strogatz}(1998)}]{watts1998collective}%
  \BibitemOpen
  \bibfield  {author} {\bibinfo {author} {\bibfnamefont {Duncan~J}\
  \bibnamefont {Watts}}\ and\ \bibinfo {author} {\bibfnamefont {Steven~H}\
  \bibnamefont {Strogatz}},\ }\bibfield  {title} {\enquote {\bibinfo {title}
  {Collective dynamics of ‘small-world’ networks},}\ }\href@noop {}
  {\bibfield  {journal} {\bibinfo  {journal} {Nature}\ }\textbf {\bibinfo
  {volume} {393}},\ \bibinfo {pages} {440} (\bibinfo {year}
  {1998})}\BibitemShut {NoStop}%
\bibitem [{\citenamefont {Chen}\ \emph {et~al.}(2016)\citenamefont {Chen},
  \citenamefont {Schonger},\ and\ \citenamefont {Wickens}}]{chen2016otree}%
  \BibitemOpen
  \bibfield  {author} {\bibinfo {author} {\bibfnamefont {Daniel~L}\
  \bibnamefont {Chen}}, \bibinfo {author} {\bibfnamefont {Martin}\ \bibnamefont
  {Schonger}}, \ and\ \bibinfo {author} {\bibfnamefont {Chris}\ \bibnamefont
  {Wickens}},\ }\bibfield  {title} {\enquote {\bibinfo {title} {o{T}ree---{A}n
  open-source platform for laboratory, online, and field experiments},}\
  }\href@noop {} {\bibfield  {journal} {\bibinfo  {journal} {Journal of
  Behavioral and Experimental Finance}\ }\textbf {\bibinfo {volume} {9}},\
  \bibinfo {pages} {88--97} (\bibinfo {year} {2016})}\BibitemShut {NoStop}%
\end{thebibliography}%

\FloatBarrier
\newpage

\begin{center}
	\textbf{\large Supplemental Materials: Trading in Complex Networks}
\end{center}
\setcounter{equation}{0}
\setcounter{figure}{0}
\setcounter{table}{0}
\setcounter{page}{1}
\makeatletter
\renewcommand{\theequation}{S\arabic{equation}}
\renewcommand{\thefigure}{S\arabic{figure}}
\renewcommand{\thetable}{S\arabic{table}}
\renewcommand{\bibnumfmt}[1]{[S#1]}
\renewcommand{\citenumfont}[1]{S#1}

\section{Experimental setup}  \label{sec:ExperimentalSetup}

\subsection{Volunteers’ recruitment and experimental platform} \label{subsec:RecruitmentAndSessions}

The experiment was carried out with volunteers chosen among the pool of volunteers of the EU project IBSEN (ibsen-h2020.eu). This pool consists of more than 24,000 volunteers
(24,210 at the time of the experiment) from different countries to perform social experiments. We restricted the call to volunteers residing in Madrid or Zaragoza (10,408 volunteers), where the experiments were performed. Following the call for participation, we selected 144 volunteers,
whose demographic data are shown in Table \ref{tableS1}.
In order to satisfy ethical procedures, all personal data about the participants were anonymized and treated as confidential.

\setlength{\tabcolsep}{3pt}

\begin{table}[h]
	\centering

	\begin{tabular}{lrrrrrrr}
		Network & Session & Number of participants &  Mean age (SD) & Range age & women & City & Date\\
		\toprule
		R 50 & 1 & 48 & 32.42 (11.84) & 21-64 & 64.58\% & Madrid and Zaragoza & 29 JUN 2017\\
		R 26 & 1 & 24 & 39.54 (12.50) & 19-61 & 70.83\% & Zaragoza & 14 NOV 2017 \\
		R 26 & 2 & 24 & 34.29 (15.32) & 20-66 & 62.50\% & Zaragoza & 24 APR 2018 \\
		SW 26 & 1 & 24 & 22.67 (2.39) & 19-31 & 20.83\% & Madrid & 31 OCT 2017 \\
		SW 26 & 2 & 24 & 34.29 (15.32) & 20-66 & 62.50\% & Zaragoza & 25 APR 2018\\
		\botrule
		Aggregated & & 144 & 33.55 (13.76) & 19-72& 57.64\% \\
		\botrule
	\end{tabular}
	
	\caption{\textbf{Experimental sessions}. Demographic data of the participants.
	R50 stands for a random network of 50 nodes, R26 for a random network of 26 nodes
	and SW26 for a small-world-like network of 26 nodes. For the R26 and SW26 networks, two experimental sessions were carried out.
	The R50 session was held jointly in Madrid and Zaragoza, the rest of sessions were
	performed either in Madrid or in Zaragoza.}
	\label{tableS1}
\end{table}

\setlength{\tabcolsep}{10pt}

The experiment was run using an application based on the oTree platform \cite{chen2016otree}. At the beginning of each session, all participants were located at a desktop computer spot. When everybody had read the instructions, the first series began. Each series lasted 15 rounds —approximately 15 minutes. At the end of the fourth series, the experiment finished and all participants received their earnings plus a show-up fee of 5 euros. Total earnings ranged from 5 to 82 euros (mean was 18.4), including the show-up fee, and the sessions lasted between 70 and 85 minutes (mean was 77), including instructions reading. The experiment is described in the section `Experimental setup' of the main text. Both the sessions and the networks are described in the section `Materials and methods' of the main text. 

\subsection{Description of networks}

All the networks used in the experiment are generated through the Watts-Strogatz \cite{watts1998collective} algorithm with different probabilities $p$ of rewiring ($p = 0.1$ for the small-world-like network and $p = 1$ for the random networks). 
In any given treatment (R26, SW26, and R50), the same network was used across all series and sessions. However, the selection of source-destination pairs is generated randomly at the beginning of each series such that the shortest path between the two nodes is of distance at least diameter - 2 (\textit{resp.} - 1) for the R26 and SW26 (resp. R50) networks (as a means to prevent uninteresting scenarios with very short distance between source and destination nodes). 
Mean degree is $3$ for the 26 nodes networks and $4$ for the network of 50 nodes. Fig. \ref{FigureExperimentalNetworks}
depicts the structural representation of the network as viewed by the subjects in each corresponding treatment: Random Network of 26 nodes (left), Small-World-Like Network of 26 nodes (center) and Random Network of 50 nodes (right).

\begin{figure}[H]
	\centering
	\includegraphics[width=0.32\textwidth]{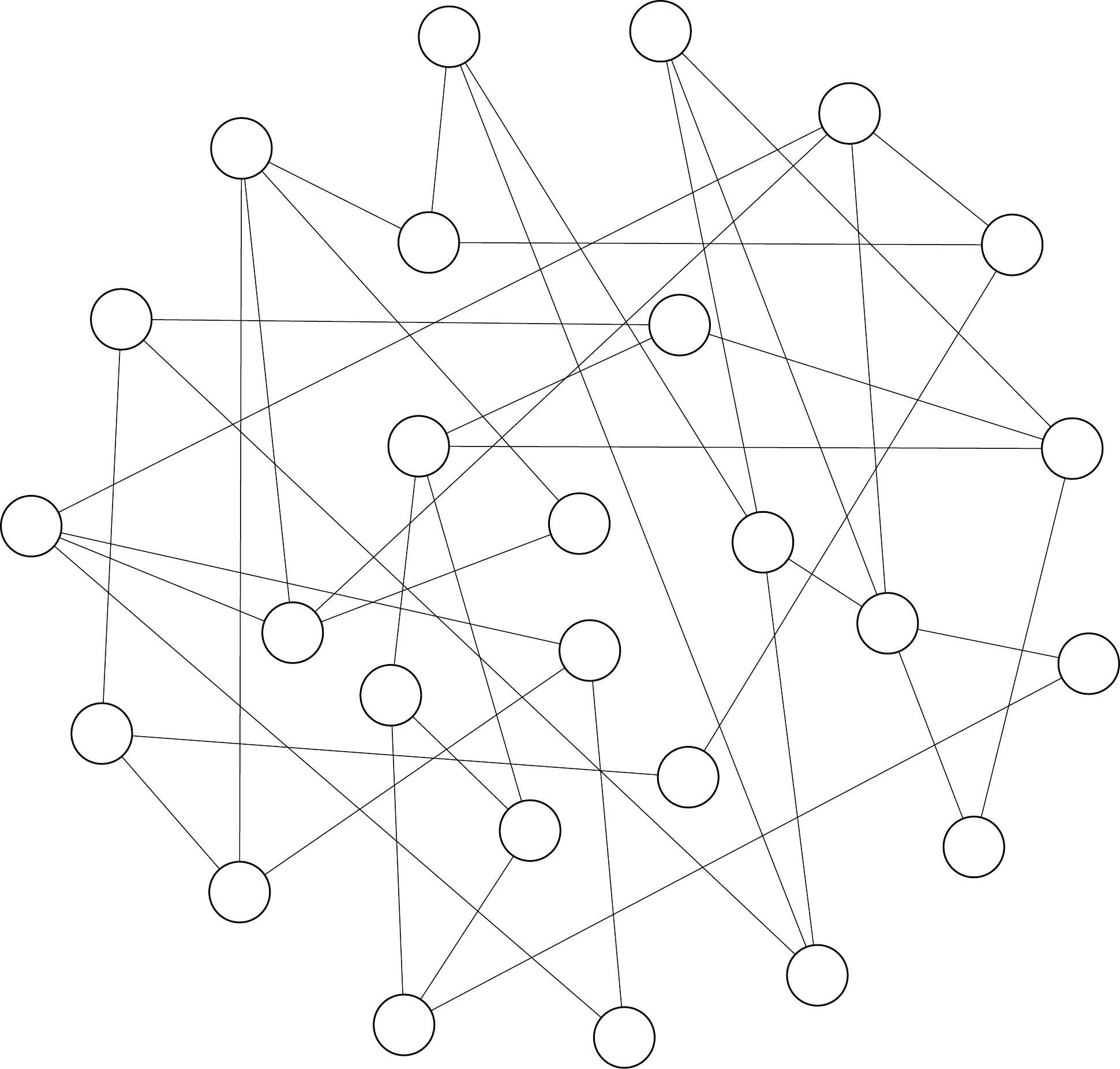}
	\includegraphics[width=0.32\textwidth]{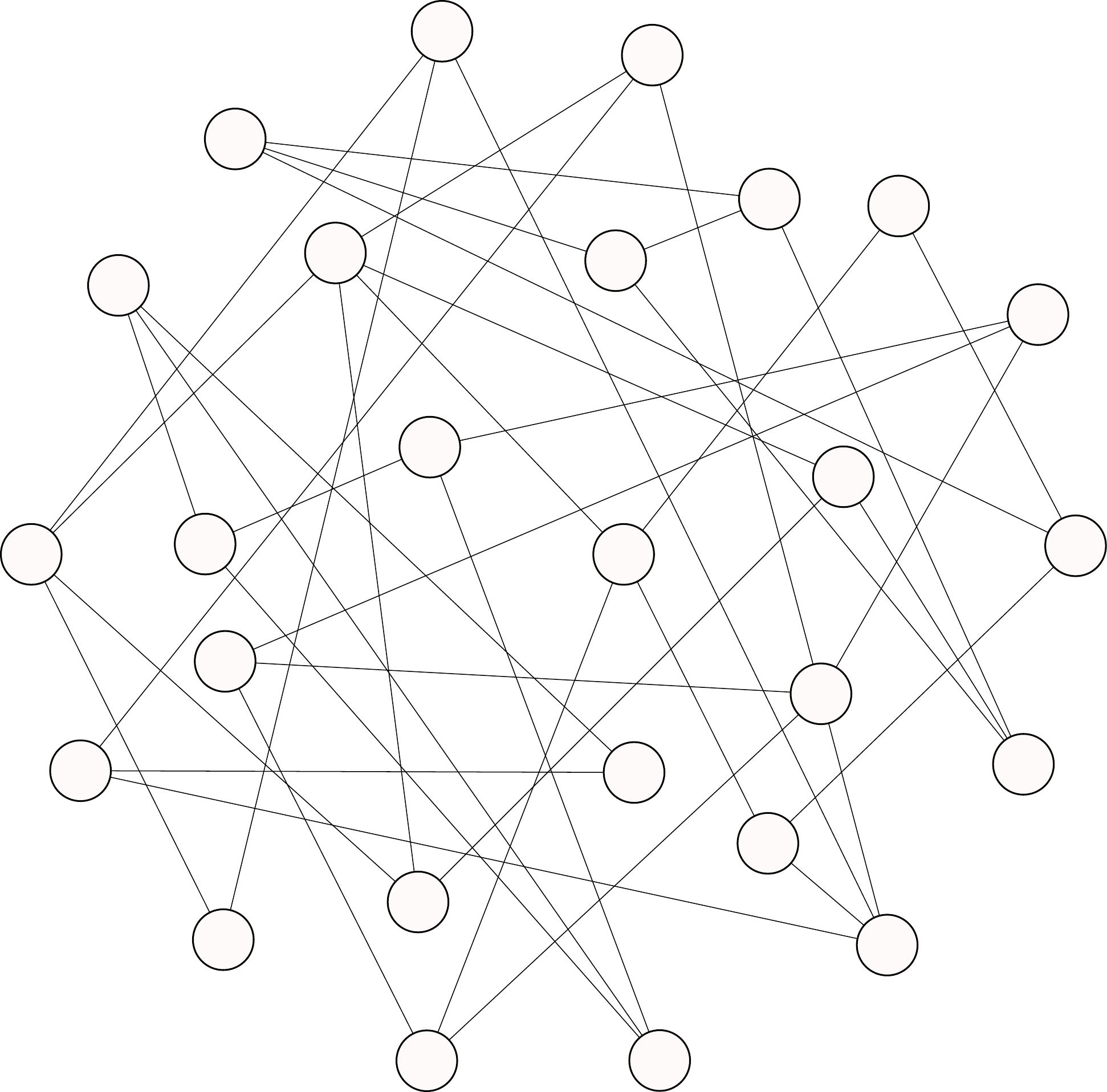}
	\includegraphics[width=0.32\textwidth]{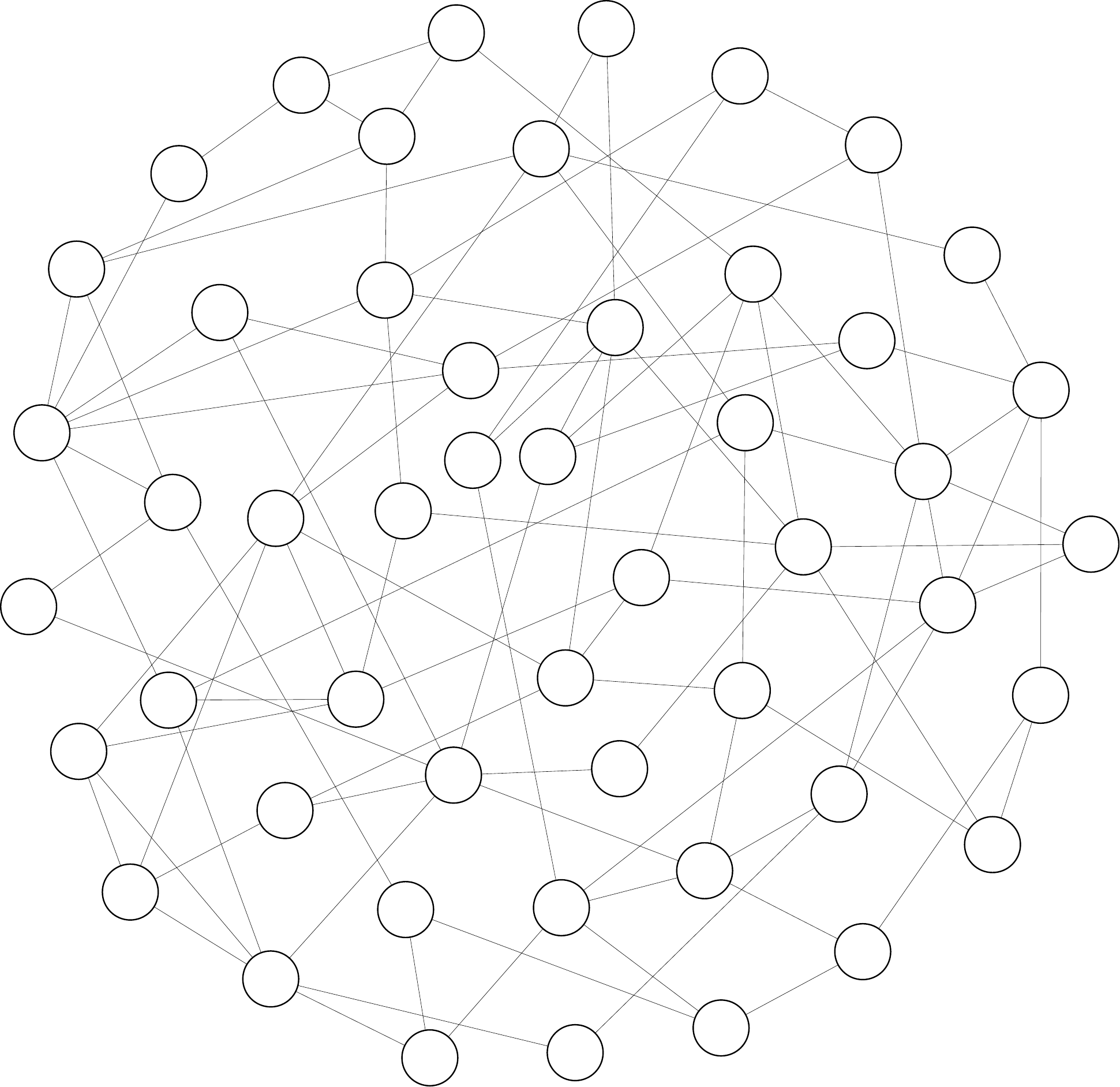}
	\caption{Networks used in the experiments: 26-nodes Random Network (left), 26-nodes Small-World-Like network (center) and 50-nodes Random Network (right).}
	\label{FigureExperimentalNetworks}
\end{figure}

\subsection{Instructions for players}

Since there were three possible networks
(random networks of 50 and 26 nodes and small-world-like network of 26 nodes), there were three different instructions that differentiated only in this single aspect. In what follows a translation of the original Spanish instructions (available upon request) for the 26-nodes random network is included.

\vspace{0.8cm}
\textbf{Instructions}
\vspace{0.2cm}

Thank you for participating in this experiment, that is part of a research project in which we try to understand how individuals make decisions. You re not expected to behave in any particular way. At this moment the experiment begins. Please keep quiet until the end, turn your cell	phone off, and remember that the use of any material foreign to the experiment is not allowed (including pen, pencil or paper).

Your earnings will depend on your own decisions and those of the other participants. Additionally, you will receive 5 \euro{} for participating in the experiment until the end.

Please keep quiet during the experiment. If you need help, raise your hand and wait to be assisted. Please do not ask any question aloud.

You participate along with other people with whom you interact according to the rules explained below. The session lasts about an hour and a half. The following instructions are the same for every participant of this experiment.

Once completed the session, you will receive 5 \euro{} for participating, along with your earnings corresponding to the rounds, once converted into euros. For convenience, the total earnings are rounded up to the nearest 50 cents.

You will access the experiments after reading these instructions. When all participants have accessed, the rounds will begin.

You are going to participate in 4 experiments. Each experiment consists of 15 rounds. Before starting each experiment, all the players, you included, will be randomly located in the nodes of the network shown below. Your position in the network will be denoted with the letter `\textbf{M}' (for me). In the same way, two different nodes will be chosen as Source (\textbf{S}) and Destination (\textbf{D}) respectively. Their position in the network will be denoted with the letters `\textbf{S}'  and `\textbf{D}'. \textbf{All the players will remain in the same position during each experiment of 15 rounds.} In the same way, the source and destination will remain in the same position throughout each experiment of 15 rounds. The players will play the role of intermediaries. A good must be transported from \textbf{S} to \textbf{D} generating a benefit of 100 tokens for all players involved (\textbf{S}, \textbf{D}, and all nodes in the path between them). Intermediaries (that is, the players) simultaneously have to post the fraction of these 100 tokens they would like to charge if selected, which must be between 0 and 100 tokens. \textbf{You will have 60 seconds to post your price.} If you do not post a price, the computer will decide for you: please do not run out your time and make your own decision.

This is the screen you will see in the first round (this screenshot is only an example):

\begin{figure}[H]
	\centering
	\includegraphics[width=0.85\textwidth]{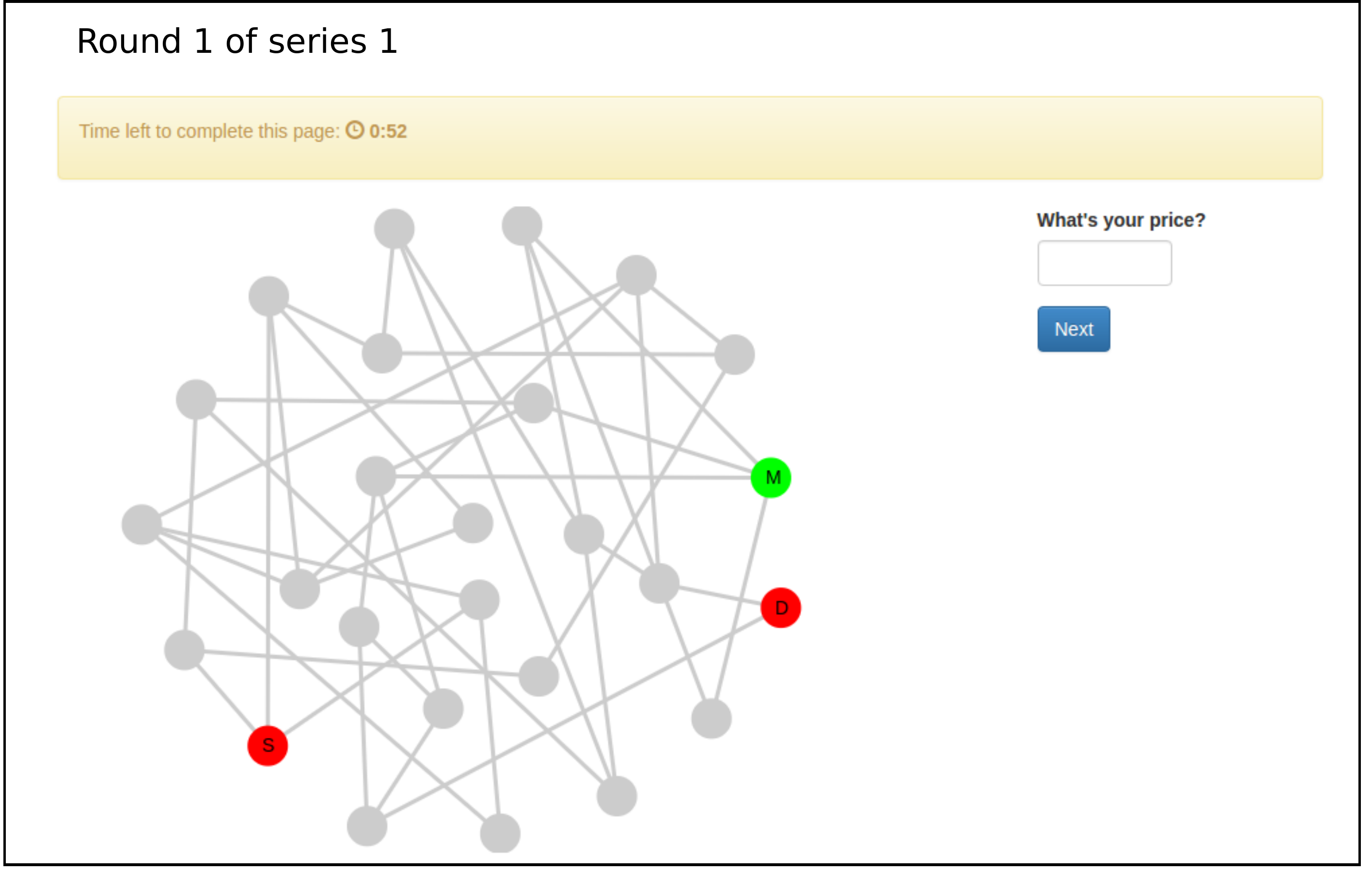}
	\label{fig:screenshot1}
\end{figure}

The sum of prices along any given path between \textbf{S} and \textbf{D} determine a total cost. Once all the intermediaries have posted a price, the cheapest path (with lowest total cost) from \textbf{S} to \textbf{D} will be selected. If the total cost of the cost cheapest path is less than or equal to 100 tokens, the good will be taken from \textbf{S} to \textbf{D}. Otherwise, that is, if all the paths from \textbf{S} to \textbf{D} cost more than 100 tokens, there will be no deal and no value will be generated. Ties are broken randomly, that is, if there are more than one cheapest path, one of them will be selected at random.

Your payoff in this round will be:

a) If you are located on the selected cheapest path, you will receive your price as payoff.
b) Otherwise, that is, if you are not on the selected path, you will not receive any payoff in that round.
S and D will receive, equally distributed, the rest of the 100 tokens.

From the second round on, you will be informed about whether there was a deal in the previous round, and if so what was the selected cheapest path, and the costs of this path. You will also be informed about the cheapest path through your node regardless of whether this was the selected cheapest path. The selected path will be highlighted by a dashed red line, while the cheapest path through your node will be highlighted by a blue solid line. Note that the cheapest path through your node may contain loops, i.e., it may pass more than one time through some nodes. With this information on the screen, you must set a price for the current round. At the end of each experiment, the positions of all players, the source, and the destination will be randomly reassigned, and a new experiment of 15 rounds will begin.

This is the screen you will see in the subsequent rounds (this screenshot is only an example):

\begin{figure}[H]
	\centering
	\includegraphics[width=0.85\textwidth]{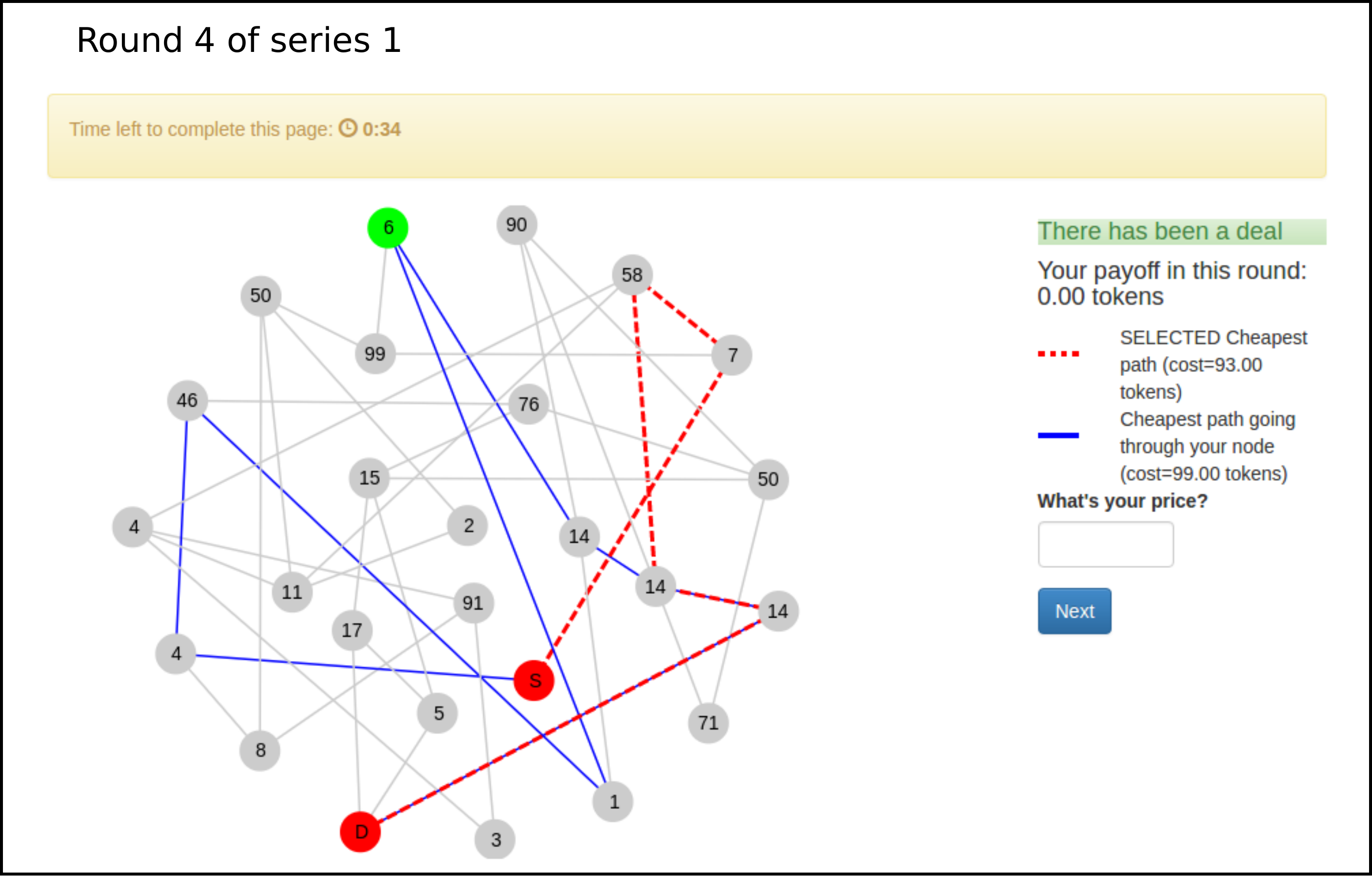}
	\label{fig:screenshot2}
\end{figure}

Please, click the below NEXT button to start:

[NEXT]

\section{Additional results}

Here, we present additional results aimed to support the findings shown in the main text.

\subsection{50 nodes random network}

The robustness of the findings
shown in the main text against the size and connectivity of the network
are explored through an additional experimental session in a
random network of 50-nodes with $\langle k \rangle =4$. Table
\ref{tableS2} puts together the results corresponding to the
random network with 50 nodes and those corresponding to 26-nodes networks.
Although the larger network shows lower costs than the smaller ones (actually, for the 50-nodes network intermediation rents are close to zero), both
the correlation between payoffs and SD-betweenness (Fig. \ref{fig:si.ppnet}) 
and the behavioral rule (Fig. \ref{fig:changeInPrice})
are verified in the network of 50 nodes, as will be discussed in the next paragraphs.

\begin{table}[h]
	\centering

	\begin{tabular}{lrrrrrr}
		\hline
		network & efficiency & price &  price in CP & cost & profit & length \\
		\toprule
		R 50 & 1 & 7.85 & 1.18 & 5.67 & 0.12 & 5.78 \\
		R 26 & 0.97 & 11.34 & 5.49 & 28.33 & 1.10 & 6.26 \\
		SW 26 & 0.68 & 18.10 & 13.16 &  76.52 & 2.38 & 7 \\
		\botrule
	\end{tabular}

	\caption{\textbf{Experimental results.} Efficiency (fraction of rounds
	in which the cheapest path cost was
	equal to or less than the threshold), and mean values of the price, price on the cheapest path, cost of the cheapest path,  profit, and cheapest path length
	for each one of the three studied
	networks: random networks with 50 and 26 nodes (R 50, R 26) and small-world network with 26 nodes
	(SW 26).}
	\label{tableS2}
\end{table}


Fig. 2 of the main text shows the prices, payoffs and frequencies on the cheapest path for all realizations. This result indicates a correlation between profits and the source-destination betweenness ($sd_{\infty}$). To ensure that this result is not spurious, Fig. \ref{fig:si.ppnet} shows how those variables correlate with the centrality in each network. It is clear that the same behavior is maintained when we look at the results for each network. 

\begin{figure}[ht]
	\centering
	\includegraphics[width=0.8\textwidth]{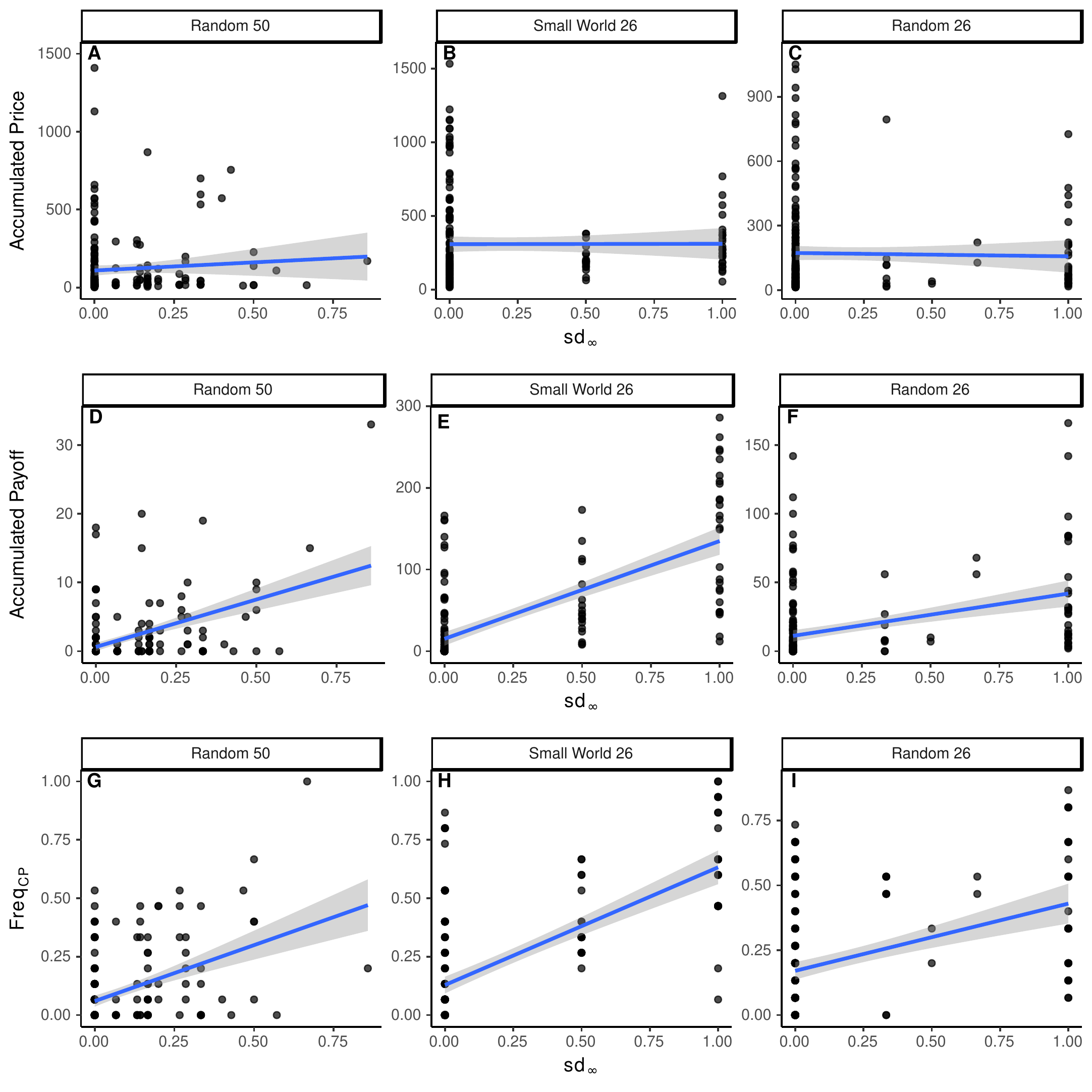}
	\caption{\textbf{SD-betweenness determines payoffs but not
		posted prices.}
		\textbf{A-I}: Accumulated price (\textbf{A-C}), accumulated payoff (\textbf{D-F}) and frequency on the cheapest path (\textbf{G, I}) of participants during a series of 15 rounds as a function of 
		the  SD-betweenness $sd_\infty$ for the Random 50 (\textbf{A,D,G}), Small World 26 (\textbf{B,E,H}), and Random 50 networks(\textbf{C,F,I}).}
	\label{fig:si.ppnet}
\end{figure}

As an extension of Fig. 3 of main text, Fig. \ref{fig:changeInPrice}A shows the mean change in price for the cases
where the participant was or was not on the cheapest path in the previous round, while Fig. \ref{fig:changeInPrice}B shows the probabilities to increase and to decrease the posted price conditioned to have been (Y) or
not (N) on the cheapest path. In this figure, the results corresponding to
the random network with 50 nodes have been added showing that,
within the limitations of the current experiment,
the behavioral rule according to which players increase their price if they were on the cheapest path in the previous round and decrease it otherwise 
is robust against the size and connectivity of the network.

\begin{figure}
	\centering
	\includegraphics[width=0.7\linewidth]{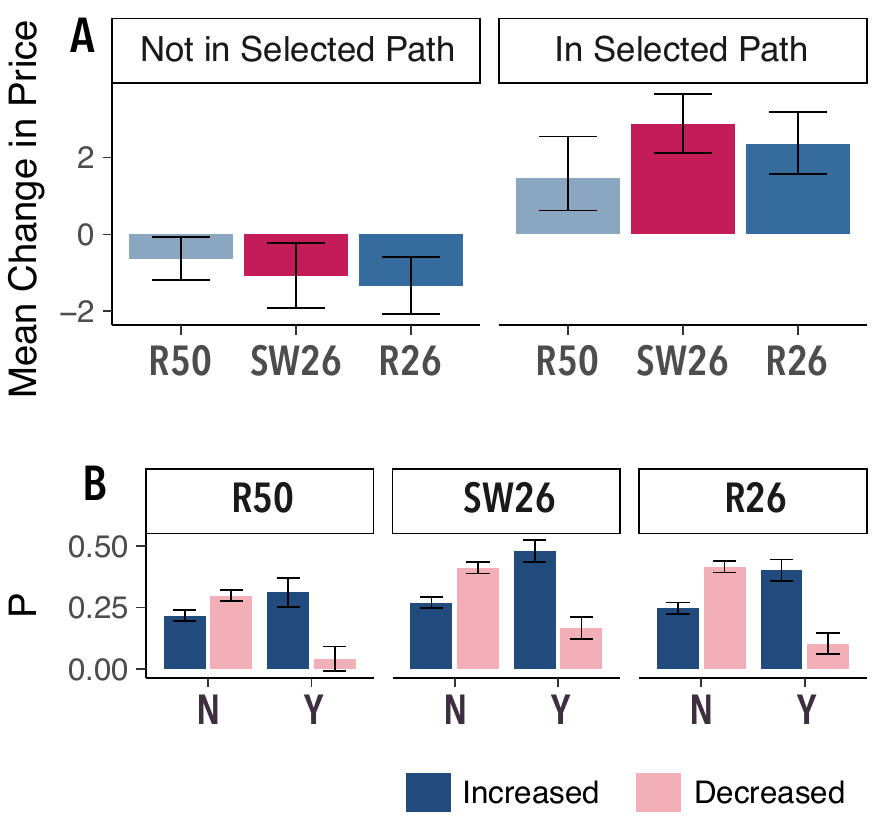}
	\caption{\textbf{Extension of figure 3 of the main text with the 50-nodes random network.}
		\textbf{A:} Mean changes in the posted price for the studied networks: random network of 50 nodes (R50) and 26 nodes (R26), and small-world network of 26 nodes (SW26). The panel discriminates the cases in which the participant was (right) from those in which she was not (left) on the selected cheapest path in the previous round.
		\textbf{B:} Probability to increase
		(red) and to decrease (blue) the posted price conditioned to have been (Y) or not (N) on the selected cheapest path, for each one of the
		studied networks.
		The error bars represent the 95$\%$ C.I.
	}
	\label{fig:changeInPrice}
\end{figure}

\subsection{Additional experimental results: Evolution of costs and prices}

Fig. \ref{fig:chpt.evo}A shows the evolution of the cheapest path cost for each network. As described in the main text, the costs are significantly higher in the Small-World network than in the Random Network. Among the random networks considered, the costs are lower in the larger network (50 nodes) than in the smaller one (26 nodes). These higher costs in the SW network entail a lower efficiency, as can be seen in table \ref{tableS2}. 
To deepen this issue, Fig. \ref{fig:chpt.evo}B represents the evolution  of the mean price of participants on the cheapest path for each network. Note that the pattern is consistent with Fig. \ref{fig:chpt.evo}A, pointing out that the previous finding was not a consequence of different cheapest path lengths. To complement these results, Fig. \ref{fig:chpt.evo}C and \ref{fig:chpt.evo}D show, respectively, the evolution of the mean and median price posted by all the participants. As has been discussed in the main text, the effect of the network topology on prices (Fig. \ref{fig:chpt.evo}C, \ref{fig:chpt.evo}D) is not as clear as it is on costs and payoffs 
(Fig.  \ref{fig:chpt.evo}A, \ref{fig:chpt.evo}B) since the topology mainly affects the probability for the nodes to be on the cheapest path more than the posted prices.

\begin{figure}[h]
	\centering
	\includegraphics[width=0.8\textwidth]{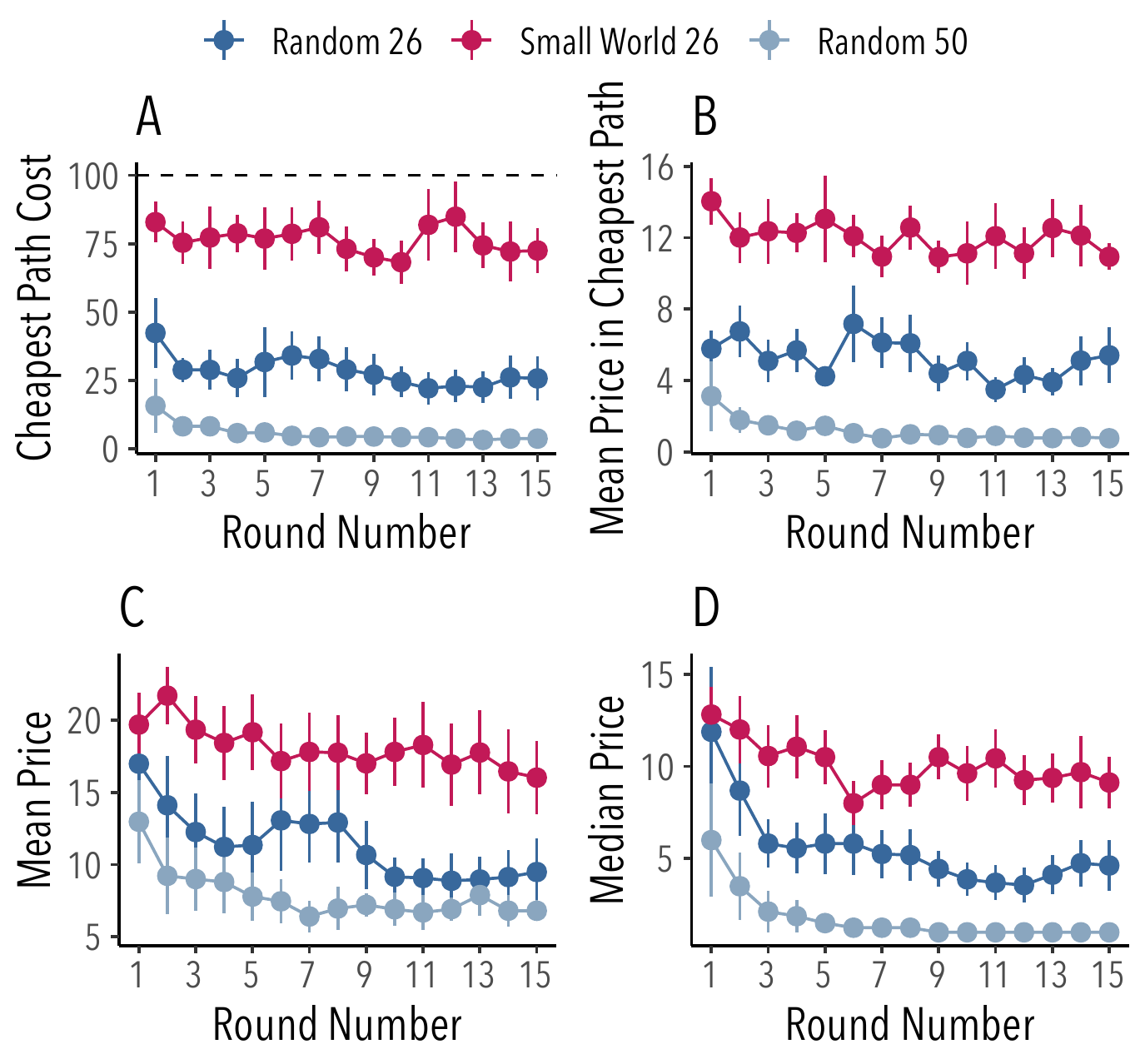}
	\caption{\textbf{Evolution of costs and prices for each experimental network.} Cheapest path cost (\textbf{A}), mean price of nodes on the cheapest path (\textbf{B}), mean price (\textbf{C}), and median price (\textbf{D}) as a function of the round number. Each series of points corresponds to a given network: Random Network with 50 nodes (gray), Small-World with 26 nodes (orange), and Random Network with 26 nodes (blue). The error bars represent 1.96 $\times$ SEM.}
	\label{fig:chpt.evo}
\end{figure}

\subsection{Additional model results: Final costs and network metrics} \label{sec:si.cost}

\subsubsection{Node-disjoint paths} \label{sec:si.cost.M}

In the game, nodes in a path between the source and the destination constitute a group, and all such groups can be seen as competing to represent the cheapest path. At the same time, individuals within those groups try to maximize their own profit, and as such are competing with other group members. Therefore, with a large number of paths there will be a large number of groups competing to be the cheapest path, with one caveat: paths do not constitute disjoint sets of nodes, which implies that nodes compete concurrently in a large number of groups. These conditions make the assessment of a good descriptor of the competition in a network a non-trivial task. Nonetheless, we can provide a lower bound estimator based on the minimum number of independent groups that can compete with each other using the number of node-disjoint paths\footnote{There is also an important practical fact to take into consideration: metrics based on counting all paths are unfeasible for relatively large and dense networks, as the number of paths grows exponentially.  Fortunately, computing the maximum number of independent paths is reducible to the maximum flow problem, thus it can be computed in polynomial time.}.

The results shown in Fig. 6 of the main text show that,
when considering the behavioral rule found, the number $M$ of node-disjoint paths is
a particularly good indicator of the final trading costs. With more independent paths, more
groups of different nodes can coordinate resulting in a cheaper cost.
To illustrate the mechanism behind this relationship, we propose the following simple scenario wherein a formal relationship between the two variables can be demonstrated:
i) all the node-disjoint paths are shortest paths
and ii) all the nodes start with the same price. For this kind of scenario, we can see a clear relation between costs and $M$, as shown by Lemma \ref{thr:t1}. Indubitably, this lemma does not provide us a rule for how costs will change in all the possible networks. Nonetheless, it provides a useful insight into how the competition between paths for the cheapest path should be related to $M$.

\begin{lemma}\label{thr:t1}
	Let us consider a graph $G$, a source $S$, a destination $D$, and identical initial prices across all nodes. Let us assume that nodes increase their posted prices by $\sigma$
	if they were located in the previous round in the cheapest path, otherwise decrease it by $\rho$. If there are $M$ node-disjoint shortest paths of the same length
	between $S$ and $D$, the cheapest path cost will increase indefinitely if only
	if $$\frac{\sigma}{\rho}>(M - 1)$$
\end{lemma}

\begin{proof}
	
	Let us consider an enumeration of the disjoint paths from $S$ to $D$: $p_1,p_2,\ldots,p_M$.
	Let $x$ be the initial posted price for all the nodes.
	Initially, a path (without loss of generality, $p_1$) will be selected, and nodes belonging to path $p_1$ will increase their price from $x$ to $x+\sigma$, while nodes belonging to the rest of paths ($p_2,p_3,\ldots,p_M$) will decrease their price to $x - \rho$. In the subsequent steps, the rest of the paths ($p_2,p_3,\ldots,p_M$) will be selected until all the $M$ paths will have been selected. At step $M$,
	the selected path will have a cost per node of $x - (M-1) \rho$, thus, its nodes will increase their price to  $x+\sigma - (M-1) \rho$, which will be higher than $x$ only if $\sigma > (M-1)\rho$. Note that, by the same reasoning, at step $M+1$ all the nodes
	located in {$p_1,\ldots, p_M$} will have a cost of $x+\sigma - (M-1) \rho$. Therefore,
	the cost will increase indefinitely if and only if $\sigma > (M-1)\rho$. 
\end{proof}

\subsubsection{Average Path Length and Clustering Coefficient} \label{sec:si.cost.pl}

Fig. 6 from the main text shows how trade costs scale with the average path length of the network. 
This result is not a consequence of cheapest paths length differences: the mean price of nodes on the cheapest path also correlates with the average path length, as shown in Fig. \ref{fig:si.meancost.sim}.

It is well known that small world networks differ from random networks with respect to clustering and average path length in different ways. Therefore, the clustering coefficient is also a natural candidate for capturing the differences in the cheapest path cost, as it is indeed the case as shown in Fig. \ref{fig:si.fig6cc}. To check which of both observables is more connected with the network properties driving the differences in cost, we executed two linear regressions with final trade cost (i.e., on the cheapest path) as the dependent variable: one having the clustering coefficient as independent variables and one with the average path length as independent variables. The data considered are the results from simulations executed without the threshold, which are shown in Fig. 6 on the main text. As the slopes of both properties change with respect to $M$, we added an individual coefficient for each value of $M$ obtained by multiplying the network property by a Kronecker delta ($\delta$) dummy variable. The regression model with clustering coefficient ($T$) is shown in Eq. \ref{eq:regT}, and that with average path length ($L$) in Eq. \ref{eq:regL}. We restricted the data to $M \leq 3$, as for larger values final costs are mostly zero.

\begin{equation}
C_i = \beta_1T_i\delta_{M_i1} + \beta_2T_i\delta_{M_i2} + \beta_3T_i\delta_{M_i3} \label{eq:regT}
\end{equation}

\begin{equation}
C_i = \beta_1L_i\delta_{M_i1} + \beta_2L_i\delta_{M_i2} + \beta_3L_i\delta_{M_i3} \label{eq:regL}
\end{equation}

The results from the regression are displayed in Table \ref{table:coefficients} which shows that the model with average path length as regressors provides a better description of the final trade costs: the coefficient of determination 
when considering the average path length is $R^2(L)=0.79$, while the model with the clustering coefficient achieves $R^2(T)=0.57$.

\begin{table}
	\begin{center}
		\begin{tabular}{l c c }
			\hline
			& Model with Clustering Coefficient & Model with Average Path Length \\
			\toprule
			Clustering Coefficient for $M=1$ & $0.74^{***}$  &              \\
			& $(0.00)$      &              \\
			Clustering Coefficient for $M=2$ & $0.04^{***}$  &              \\
			& $(0.00)$      &              \\
			Clustering Coefficient for $M=3$ & $-0.09^{***}$ &              \\
			& $(0.00)$      &              \\
			Average Path Length for $M=1$  &               & $0.97^{***}$ \\
			&               & $(0.00)$     \\
			Average Path Length for $M=2$ &               & $0.20^{***}$ \\
			&               & $(0.00)$     \\
			Average Path Length for $M=3$  &               & $0.07^{***}$ \\
			&               & $(0.00)$     \\
			\hline
			R$^2$      & 0.57          & 0.79         \\
			Adj. R$^2$ & 0.57          & 0.79         \\
			Num. obs.  & 131148        & 131148       \\
			RMSE       & 0.66          & 0.46         \\
			\botrule
			\multicolumn{3}{l}{\scriptsize{$^{***}p<0.001$, $^{**}p<0.01$, $^*p<0.05$}}
		\end{tabular}
		\caption{\textbf{Coefficients of the statistical models.} Coefficients of the two statistical models, having the clustering coefficient as the independent variable (\textit{left}), and having the average path length as the independent variable (\textit{right}). The coefficients are standardized (centered and divided by their standard deviation). As the coefficients of clustering and average path length seem to vary with $M$, we consider it as a dummy variable. Slopes (clustering coefficient and average path length) are $M$-specific.}
		\label{table:coefficients}
	\end{center}
\end{table}

\begin{figure}[h]
	\centering
	\includegraphics[width=0.9\textwidth]{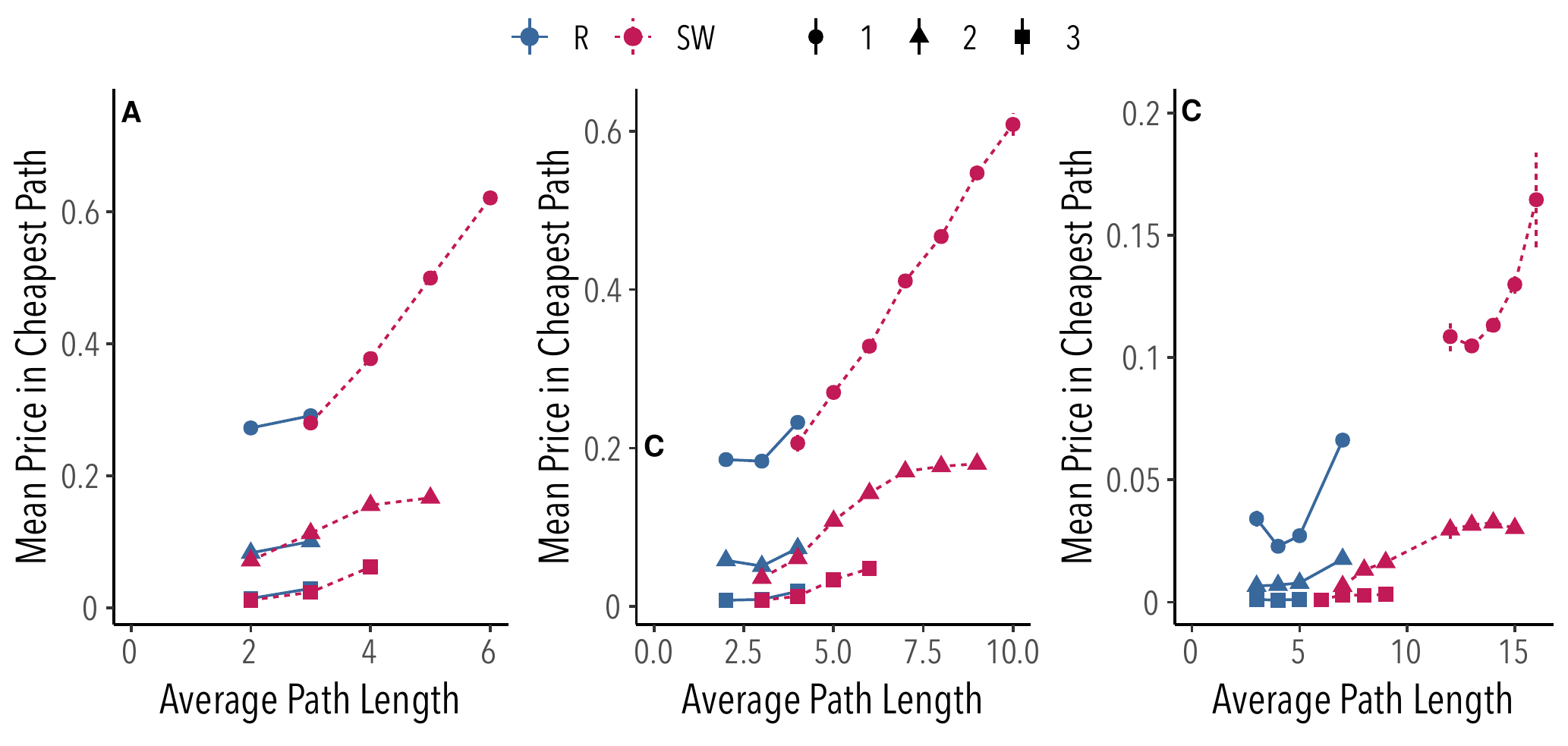}
	\caption{\textbf{Numerical results of the model.} Average price of nodes on the cheapest path after a period of $10^4$ rounds as a function of the mean path length of the network. Each panel corresponds to a different network size: 26 (\textbf{A}), 50 (\textbf{B}), and 1000 (\textbf{C}) nodes. Different colors correspond to different network models: random (\textit{red}) and small-world (\textit{blue}). Different simbols correspond to different values of the number $M$ of disjoint paths:  $M=1$ (circles), 2 (triangles), and 3 (squares).
		For each configuration, there were generated 10000 networks of each size according
		to the Watts-Strogatz algorithm \cite{watts1998collective}
		with $p=0.1$ (SW) and $p=1$ (R), and average degree from 2 to 10.
		The increment/decrement ratio was fixed to the experimental value ($\sigma/\rho=2.4$).}
	\label{fig:si.meancost.sim}
\end{figure}

\begin{figure}
	\centering
	\includegraphics[width=0.9\textwidth]{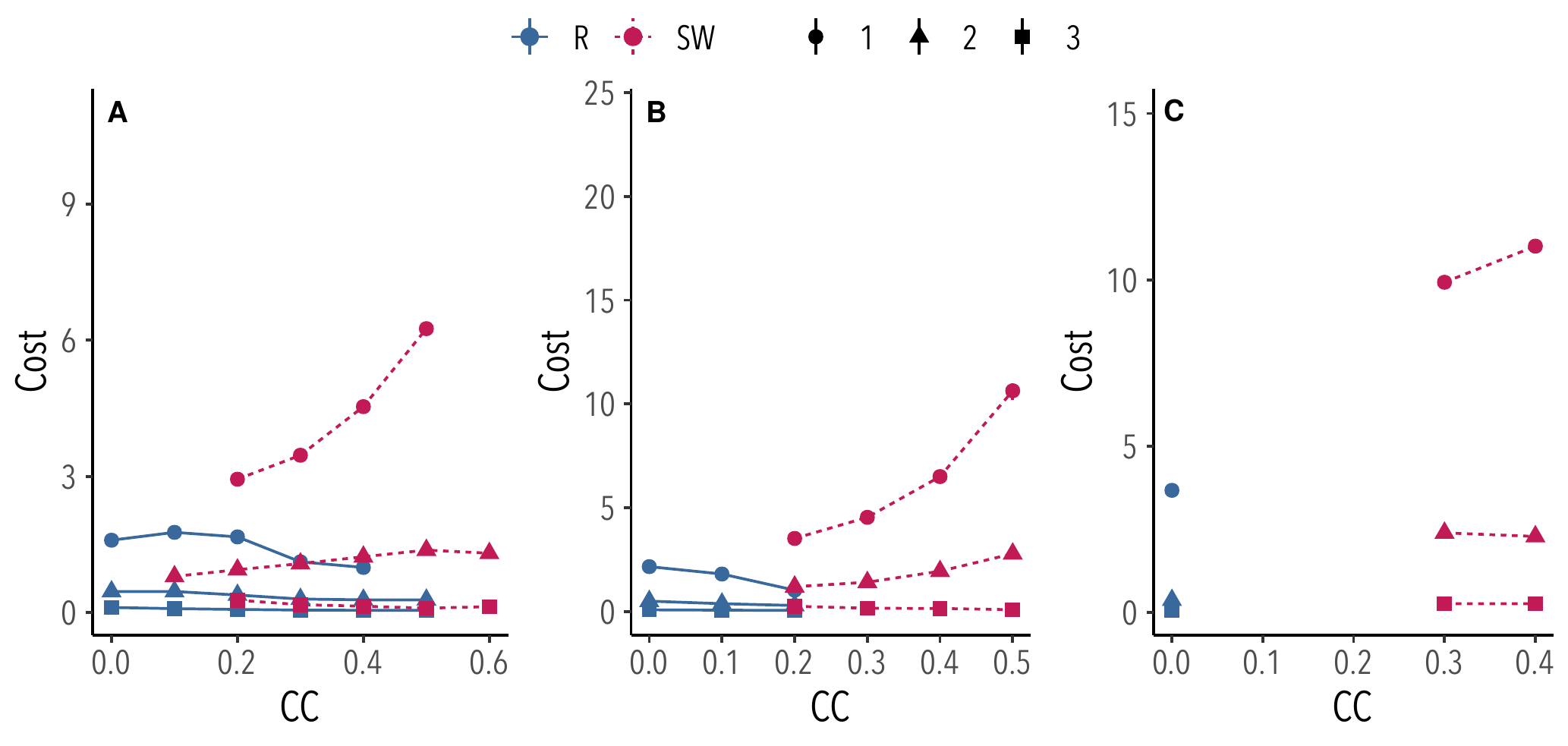}
	\caption{
		\textbf{Numerical results of the model.} Average cost of the cheapest path after a period of $10^4$ rounds as a function of the network clustering coefficient. Each panel corresponds to a different network size: 26 (\textbf{A}), 50 (\textbf{B}), and 1000 (\textbf{C}) nodes. Different colors correspond to different network models: random (\textit{red}) and small-world (\textit{blue}); while different symbols correspond to different values of the number $M$ of disjoint paths: $M=1$ (circles), 2 (triangles), and 3 (squares).
		For each configuration, there were generated 10000 networks of each size, according
		to the Watts-Strogatz algorithm \cite{watts1998collective}
		with $p=0.1$ (SW), $p=1$ (R), and average degree from 2 to 10.
		The increment/decrement ratio was fixed to the experimental value ($\sigma/\rho=2.4$).}
	\label{fig:si.fig6cc}
\end{figure}

\end{document}